\documentclass[aps,prl,reprint,longbibliography,twocolumn,superscriptaddress,nofootinbib,10pt]{revtex4-2}

\usepackage[utf8]{inputenc}
\usepackage[T1]{fontenc}
\usepackage{amsmath,amssymb}
\usepackage{graphicx}
\usepackage{braket}
\usepackage{xr}
\usepackage{hyperref}
\usepackage{bm}
\usepackage[capitalize]{cleveref}
\usepackage{color}
\usepackage{amsfonts,amsthm,dsfont}
\usepackage{booktabs}
\usepackage{cancel}
\usepackage{physics}
\usepackage{placeins}

\providecommand{\abs}[1]{\left\lvert #1\right\rvert}
\providecommand{\norm}[1]{\left\lVert #1\right\rVert}

\newtheorem{prop}{Proposition}
\newtheorem{lem}{Lemma}
\newtheorem{cor}{Corollary}

\crefname{thm}{theorem}{theorems}
\Crefname{thm}{Theorem}{Theorems}
\crefname{prop}{proposition}{propositions}
\Crefname{prop}{Proposition}{Propositions}
\crefname{lem}{lemma}{lemmas}
\Crefname{lem}{Lemma}{Lemmas}
\crefname{cor}{corollary}{corollaries}
\Crefname{cor}{Corollary}{Corollaries}

\hypersetup{
    hidelinks,
    pdftitle={Generalized Mermin Inequalities for Benchmarking Large-Scale GHZ States}
}

\graphicspath{{fig_GHZ/}}

\newcommand{\GHZ}[1]{\ket{\mathrm{GHZ}_{#1}}}
\newcommand{\MM}{\mathcal{M}}

\newcommand{\HFRC}[1][Hefei National Research Center for Physical Sciences at the Microscale and School of Physical Sciences, University of Science and Technology of China, Hefei 230026, China]{\affiliation{#1}}
\newcommand{\SHRC}[1][Shanghai Research Center for Quantum Science and CAS Center for Excellence in Quantum Information and Quantum Physics, University of Science and Technology of China, Shanghai 201315, China]{\affiliation{#1}}
\newcommand{\HFNL}[1][Hefei National Laboratory, University of Science and Technology of China, Hefei 230088, China]{\affiliation{#1}}

\newcommand{\LINKELAB}[1][LINKE lab, University of Science and Technology of China (USTC), Hefei 230027, China]{\affiliation{#1}}

\newcommand{\PKUCFCS}[1][Center on Frontiers of Computing Studies, Peking University, Beijing 100871, China]{\affiliation{#1}}
\newcommand{\PKUSCS}[1][School of Computer Science, Peking University, Beijing 100871, China]{\affiliation{#1}}

\newcommand{\SIEGEN}[1][Naturwissenschaftlich-Technische Fakult\"{a}t, Universit\"{a}t Siegen, Walter-Flex-Stra\ss e 3, 57068 Siegen, Germany]{\affiliation{#1}}
\newcommand{\INRIA}[1][CPHT, LIX, CNRS, Inria, \'{E}cole polytechnique, Institut Polytechnique de Paris, Palaiseau, France]{\affiliation{#1}}

\newcommand{\QCTek}[1][QuantumCTek Co., Ltd., Hefei 230088, China]{\affiliation{#1}}

\let\revtexmaketitle\maketitle
\let\revtexauthor\author
\let\revtexaffiliation\affiliation
\let\revtexemail\email
\let\revtexthanks\thanks

\begin{document}

\title{Generalized Mermin Inequalities for Benchmarking Large-Scale GHZ States}

\author{Jianbin Cai}
\thanks{These authors contributed equally to this work.}
\HFRC
\SHRC
\HFNL

\author{Junxiang Huang}
\thanks{These authors contributed equally to this work.}
\PKUCFCS
\PKUSCS

\author{Fynn Otto}
\thanks{These authors contributed equally to this work.}
\SIEGEN

\author{Yuan Li}
\thanks{These authors contributed equally to this work.}
\HFRC
\SHRC
\HFNL

\author{Carlos de Gois}
\SIEGEN
\INRIA

\author{Tao Jiang}
\HFRC
\SHRC
\HFNL

\author{Sirui Cao}
\HFRC
\SHRC
\HFNL

\author{Fangzheng Chen}
\LINKELAB

\author{Hao Fu}
\LINKELAB

\author{Jin Lin}
\HFRC
\SHRC
\HFNL

\author{Wei Xie}
\LINKELAB

\author{Naibin Zhou}
\HFRC
\SHRC
\HFNL

\author{Shibiao Tang}
\QCTek

\author{Xiang-Yang Li}
\HFNL
\LINKELAB

\author{Cheng-Zhi Peng}
\HFRC
\SHRC
\HFNL

\author{Xiao Yuan}
\PKUCFCS
\PKUSCS

\author{Otfried Gühne}
\email{otfried.guehne@uni-siegen.de}
\SIEGEN

\author{Ming Gong}
\email{minggong@ustc.edu.cn}
\HFRC
\SHRC
\HFNL

\begin{abstract}
Multipartite Bell tests provide a correlation-only route to benchmarking quantum processors, but their application at large scales is hindered by the rapid decay of many-body correlators under noise and exponentially many terms in conventional Bell expressions. 
Here we address these scalability obstacles by introducing a finite-setting generalized Mermin family of state-tailored Bell inequalities with analytic certification bounds, 
in which the measurement-setting number $m$ provides an additional certification dimension complementary to the system size $n$.
We show that, for the powers-of-two setting choices considered here, increasing $m$ leaves the ideal normalized multipartite quantum value unchanged while lowering the relevant classical bounds, thereby strengthening the Bell-violation ratios and yielding an improved noise-robustness scaling compared to the standard Mermin inequality.  
We test this construction experimentally on a programmable superconducting processor by preparing Greenberger-Horne-Zeilinger (GHZ) states of up to 80 qubits. Using randomized sampling for direct Bell-operator estimation, we observe Bell ratios that grow exponentially with system size, certify a nonlocality depth of 14, and show that increasing $m$ strengthens both the Bell ratio and depth certification. 
All results are obtained solely from measured correlators and analytical bounds, without readout correction, tomography, or model-based mitigation. 
Generalized Mermin inequalities therefore provide a sharper Bell benchmark for noisy large-scale GHZ states. 

\end{abstract}

\maketitle

Bell inequalities rule out local-realistic descriptions of observed correlations, with their violation witnessing Bell nonlocality--a stronger form of nonclassicality than entanglement or Einstein-Podolsky-Rosen steering~\cite{EPR35,Bell1964,brunner2014bell,horodecki2009quantum,reid2009colloquium,he2012genuine,wiseman2006steering,jones2007entanglement}. 
In multipartite systems, they can further determine how many parties participate in the observed nonclassical correlations, leading to notions such as nonlocality depth and Bell-correlation depth~\cite{bancal2009Quantifying,baccari2019bell,bernards2023bell}. 
Because these conclusions rely only on measured correlators and analytical classical bounds, multipartite Bell tests provide natural correlation-only benchmarks for large quantum processors and connect, in more idealized architectures, to device-independent (DI) cryptography, self-testing, and communication-complexity advantages~\cite{acin2007device,pironio2009device,vsupic2020self,brukner2004bells,buhrman2010nonlocality}. 

However, multipartite Bell observables are increasingly fragile under decoherence, while the number of relevant global setting combinations or correlator terms can grow rapidly with system size. Existing multipartite Bell experiments therefore either remain at modest sizes, rely on additional modeling assumptions, or focus on correlation-depth witnesses rather than direct Bell tests~\cite{monz201114,schmied2016bell,ansmann2009violation,storz2023loophole,yang2022testing,lanyon2014experimental,wang2025probing}. 
Theoretical studies have shown that this logic can be partly reversed: with suitably chosen Bell operators, large scale can become an advantage, yielding Bell ratios that grow exponentially with system size~\cite{mermin1990extreme,ardehali1992bell,guhne2008generalized,bonsel2025generating} (Fig.~\ref{fig:protocol}a). Greenberger-Horne-Zeilinger (GHZ) states provide a natural arena for this idea, due to their analytically tractable Bell values and mature state-preparation techniques on current hardware~\cite{mermin1990extreme,ardehali1992bell,bao2024creating}. On noisy NISQ-era processors~\cite{preskill2018quantum}, the standard two-setting Mermin inequality does not fully exploit the available nonclassical correlations. 
Its fixed classical threshold often limits certification efficiency, leaving room to construct stronger Bell functionals by further suppressing the classical bound for the same noisy GHZ state.


Our central idea is to strengthen Bell certification by modifying the Bell functional rather than the state preparation. We promote the number of local measurement settings $m$ to a certification parameter complementary to the system size $n$ (Fig.~\ref{fig:protocol}b), thereby incorporating a broader family of equatorial GHZ correlators into the generalized Mermin operator.
For the setting choices considered here,
increasing $m$ preserves the ideal normalized quantum value while
lowering the relevant local and grouping-model bounds, and can therefore
enhance certification for the same noisy GHZ state.
We demonstrate this mechanism on the \textit{Zuchongzhi 3.1} superconducting processor~\cite{jiang2026one} using GHZ states of up to 80 qubits, observing exponentially growing Bell-violation ratios and stronger nonlocality-depth certification with increasing $m$.
Rather than inferring Bell values indirectly from fidelity or coherence diagnostics
and correlation-depth witnesses~\cite{lanyon2014experimental, wang2025probing}, we directly estimate the generalized Mermin operator by randomized sampling of its correlators.
In the correlation-only spirit of Bell benchmarking and self-testing approaches~\cite{vsupic2020self,bonsel2025generating}, all reported estimates rely solely on measured correlators and analytical bounds, without readout correction, tomography, or model-based mitigation.

{
We perform a rigorous statistical error analysis using one-sided empirical Bernstein (EB) bounds applied to the bounded setting-level averages, tying every confidence statement directly to bounded setting averages~\cite{maurer2009empirical}. 
The present superconducting implementation is not loophole-free and therefore does not constitute a full DI Bell test~\cite{weihs1998violation,hensen2015loophole,shalm2015strong,giustina2015significant}. It nevertheless demonstrates a Bell-certification strategy that is scalable and in principle portable across hardware platforms. 
}

\begin{figure*}[t]
    \centering
    \includegraphics[width=0.92\textwidth]{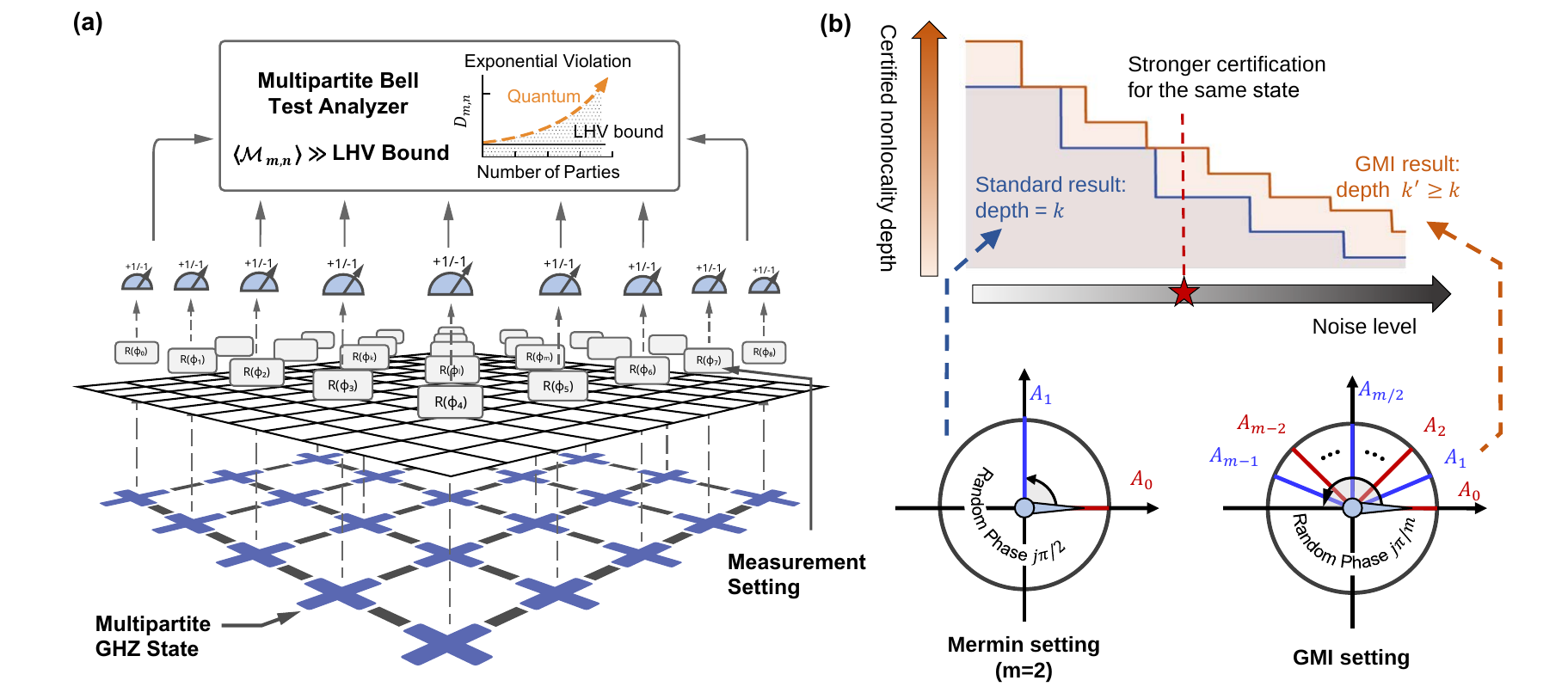}
    \caption{\textbf{Bell-test logic and measurement-space refinement.}
    (a) Local equatorial rotations $R(\phi_k)$ are applied to an $n$-qubit GHZ state to measure the generalized Mermin operator $\langle \mathcal{M}_{m,n} \rangle$. The inset shows exponential LHV violation with system size.
     (b) Increasing $m$ can tighten the relevant bounds and deepen certification for the same noisy GHZ state.
     }
    \label{fig:protocol}
\end{figure*}

\par\bigskip

\noindent\textit{Theory.} 
To certify the nonlocality of multipartite systems, we consider the $n$-qubit GHZ state, \mbox{$\ket{\text{GHZ}_n} = (\ket{0}^{\otimes n} + \ket{1}^{\otimes n})/\sqrt{2}$}. The standard Mermin inequality targets this state by combining its $X$- and $Y$-direction stabilizer observables~\cite{mermin1990extreme}. However, the underlying stabilizer group contains richer symmetries under local phase rotations. To fully exploit these correlations, we generalize this construction by introducing an arbitrary setting number $m \ge 2$, where local measurements are chosen from $m$ coplanar directions (Fig.~\ref{fig:protocol}b). Multi-setting GHZ Bell inequalities have a long history. Continuum and all-settings constructions, as well as arbitrary-setting variants in the line of Żukowski's geometrical approach, already showed that additional equatorial settings can improve the noise robustness of GHZ nonlocality~\cite{zukowski1993bell,nagata_bell_2006}.
However, large-scale experimental certification remains challenging, as it requires a finite-setting framework with analytical local and grouping-model bounds whose Bell values can be estimated without measuring exponentially many correlators.
Here we use this perspective in a finite-setting, experimentally sampled form. Our theoretical contribution is to derive and organize the local and grouping-model bounds for the normalized generalized Mermin operator used below, making the family directly applicable to Bell-ratio and nonlocality-depth certification at large $n$. We also use a biseparable quantum benchmark for the same normalized operator as a secondary GME-oriented benchmark. Specifically, our generalized Mermin operator is defined as:
\begin{equation}
    \MM_{m,n}
    =
    \frac1{m^{n-1}}\sum_{\substack{\bm{x}\in\{0,\ldots,m-1\}^n\\ s(\bm{x})\equiv 0 \,(\mathrm{mod}\, m)}}
    (-1)^{s(\bm{x})/m}
    \bigotimes_{j=1}^n A_{x_j}^{(j)},    \label{eq:gmi_main}
\end{equation}
with 
$s(\bm{x})=\sum_{j=1}^n x_j$.

To realize this Bell expression on GHZ states, we choose the local equatorial observables $A_x = \cos\!\left(\frac{\pi x}{m}\right)X+\sin\!\left(\frac{\pi x}{m}\right)Y,
    \, x\in\{0,\ldots,m-1\}$. 
For $m=2$, it reduces to the standard Mermin inequality up to normalization~\cite{mermin1990extreme}. 
For GHZ states, 
each allowed correlator satisfies $\langle\bigotimes_{j=1}^n A_{x_j}^{(j)}\rangle=(-1)^{s(\bm{x})/m}$, so every signed term in $\MM_{m,n}$ contributes $+1$ and the normalized quantum value reaches the algebraic maximum
\begin{equation}
    Q_{m,n} = \bra{\mathrm{GHZ}_n}\MM_{m,n}\ket{\mathrm{GHZ}_n}=1.
\end{equation}

The nontrivial question is how the classical benchmark scales with $m$. For powers of two, the local hidden-variable bound $C_{m,n}$ admits a closed form for even $n$, which covers the majority of the instances considered here (the corresponding odd-$n$ expression is given in the Supplementary Material),
\begin{equation}
    C_{m,n}
    =
    \frac{2}{m^n}
    \sum_{j=0}^{m/2-1}
    \frac{1}{\sin^{n}\!\left(\frac{(2j+1)\pi}{2m}\right)},
    \qquad
    m=2^r.
    \label{eq:local_bound_main}
\end{equation}
For large system sizes $n$, the bound decays exponentially as $C_{m,n}\propto\alpha_m^n$. Its exponential base $\alpha_m$ decreases with the number of measurement settings $m$, from $\alpha_2=1/\sqrt{2}\approx0.707$ to, for example, $\alpha_8\approx0.641$, thereby yielding improved noise-robustness scaling. This directly translates into a larger exponential scaling factor for the Bell-violation ratio:  $D^{\mathrm{id}}_{2,n}=1/C_{2,n}\propto1.414^n$ for the standard Mermin inequality, compared with $D^{\mathrm{id}}_{8,n}\propto1.561^n$ for $m=8$.




The same structure yields depth-sensitive classical benchmarks. If the $n$ parties are partitioned into groups of size at most $k$ and arbitrary classical communication is allowed only within each group, then the bound is
\begin{equation}
    \langle \MM_{m,n}\rangle_{k\text{-prod}}
    \le
    C_{m,\lceil n/k\rceil}.
    \label{eq:depth_main}
\end{equation}
Violating the right-hand side excludes all grouping-model explanations with maximal group size $k$ and certifies nonlocality depth at least $k+1$~\cite{bancal2009Quantifying,baccari2019bell,bernards2023bell}. We also compare selected data against the following biseparable quantum benchmark for the normalized generalized Mermin operator, following the multisetting extended-parity-game bound of P\'al and V\'ertesi
~\cite{pal2011multisetting},
\begin{equation}
    \langle \MM_{m,n}\rangle_{\mathrm{bisep}}
    \le
    \frac{1}{m\sin(\pi/2m)},
    \label{eq:bisep_main}
\end{equation}
which is useful as a smaller-scale entanglement benchmark.

To estimate $\langle\MM_{m,n}\rangle$ directly from measured correlators while avoiding exhaustive enumeration of its $m^{n-1}$ terms, which grow exponentially with $n$, we sample $N$ settings uniformly at random and acquire $M$ shots per setting. Averaging the $M$ shots for each setting yields a bounded random variable $S_i\in[-1,1]$, and the resulting estimator is $\hat{\MM}_{m,n}=\sum_{i=1}^{N} S_i/{N}$.
Because the sampled settings are randomized throughout the run, slow hardware drift is not systematically tied to any particular Mermin term.  For the concentration analysis, we treat the settings-level averages $S_i$ as independent bounded observations, while allowing their distributions and expectations to vary across the run. 
Accordingly, we take $N$, rather than the total shot count $K = NM$, as the statistical sample size. This is a conservative choice that does not rely on within-setting shot independence. 
Importantly, a large Bell violation does not inherently guarantee a high-confidence rejection of LHV models, as the statistical strength of the rejection remains heavily constrained by finite-sample fluctuations~\cite{van2005statistical, jungnitsch2010increasing}. 
To establish a rigorous mapping to strict statistical evidence, we compute all $p$-values at this settings level using the one-sided empirical Bernstein bound of Ref.~\cite{maurer2009empirical}. 
Here the $p$-value upper-bounds the probability that the relevant null model would produce an apparent violation at least as large as the observed one purely through finite-sampling fluctuations.
Writing $\hat{\mu}=\hat{\MM}_{m,n}$, $\mu:=\mathbb{E}[\hat{\mu}]=\frac{1}{N}\sum_{i=1}^{N}\mathbb{E}[S_i]$ (which equals $\mathbb{E}[S_i]$ when the settings-level variables are identically distributed under uniform random sampling), and $\hat{V}_N$ for the empirical variance of the $S_i$, the null hypothesis $H_0:\mu\le C$ is rejected with upper bound $p$-value. 
By saturating the empirical Bernstein inequality and solving analytically for $\delta$, one obtains the upper bound
\begin{equation}
    p_{\mathrm{EB}}
    =
    \min\{1,2\exp(-x_*^2)\},
    \label{eq:eb_main}
\end{equation}
where
\begin{equation}
    x_*
    = \frac{3(N-1)}{28}
    \left(
    \sqrt{
    \frac{2\hat{V}_N}{N}
    +
    \frac{56(\hat{\mu}-C)}{3(N-1)}
    }
    -
    \sqrt{\frac{2\hat{V}_N}{N}}
    \right)
    \label{eq:eb_xs_main}
\end{equation}
for $\hat{\mu}>C$, and $p_{\mathrm{EB}}=1$ otherwise, where $C$ is the relevant null-hypothesis threshold: the applicable local, biseparable, or $k$-producible bound. 

\par\bigskip

\noindent\textit{Experimental setup.}
To test the performance experimentally, we prepare GHZ states of up to 80 qubits on the \textit{Zuchongzhi 3.1} superconducting processor. The preparation proceeds by expanding entanglement outward from a central qubit using the native chip connectivity, as shown in \cref{fig:topology}(b). The optimization path for the entangling graph is determined by an automated optimization script, ensuring a shallow and hardware-aligned compilation, which is essential once the GHZ coherence is distributed across up to 80 qubits.

\begin{figure}[t]
    \centering
    \includegraphics[width=\columnwidth]{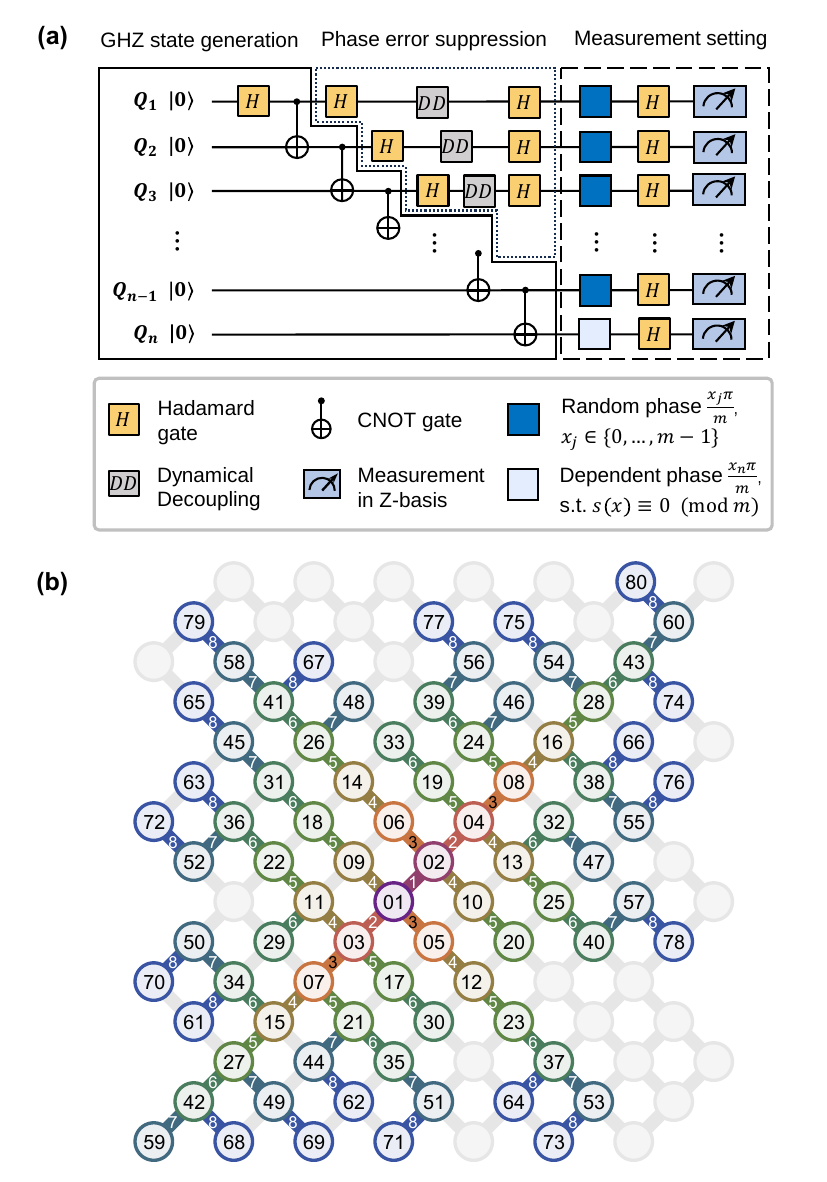}
    \caption{\textbf{GHZ preparation and hardware embedding.}
    (a) Experimental sequence for GHZ-state generation (solid), phase-error suppression (dotted), and setting-dependent measurement of \cref{eq:gmi_main} (dashed).
    (b) The accompanying chip map shows the active qubits and compiled depth for the $80$-qubit GHZ instance on the superconducting processor. The numbers on the links denote the sequence steps of entangling gates for generating the GHZ state, illustrating how the entangled region is built up step by step.
    }
    \label{fig:topology}
\end{figure}

We perform Bell certification on our processor using the protocol illustrated in \cref{fig:topology}(a). After state generation we insert a phase-protection window before the measurement rotations. The core GHZ coherence lives in the global $X^{\otimes n}$ component, so we use a Hadamard window together with dynamical-decoupling (DD) pulses to suppress phase noise during the idle interval and return the state to its original frame before the equatorial analysis~\cite{viola1998dynamical,bylander2011noise,bao2024creating}. The measurement stage applies the equatorial rotations corresponding to the number of settings $m$ and then measures in the computational basis.

In all direct comparisons across $m$, the GHZ preparation circuit, qubit subset, dynamical-decoupling sequence, and total sampling budget $K=NM$ were kept fixed; only the final equatorial analysis layer depended on $m$.

For each sampled Mermin correlator we collect $M=1500$ shots on GHZ states, and we randomize the setting order throughout the run. This randomized protocol is deliberately agnostic to any calibrated noise model: no readout correction, state tomography, or model-based mitigation enters the Bell estimates or their statistical interpretation.

\par\bigskip

\noindent\textit{Experimental results.}
We first compare the standard Mermin test ($m=2$) and a generalized Mermin test with $m=8$ as the system size increases, as shown in Fig.\ref{fig:size_scaling}. The measured generalized Mermin expectation values decay with $n$, as expected for noisy GHZ states, yet remain remarkably close for both families, whereas their experimental Bell ratios differ markedly due to the lower classical threshold of the generalized test.

\begin{figure}[t]
    \centering
    \includegraphics[width=\columnwidth]{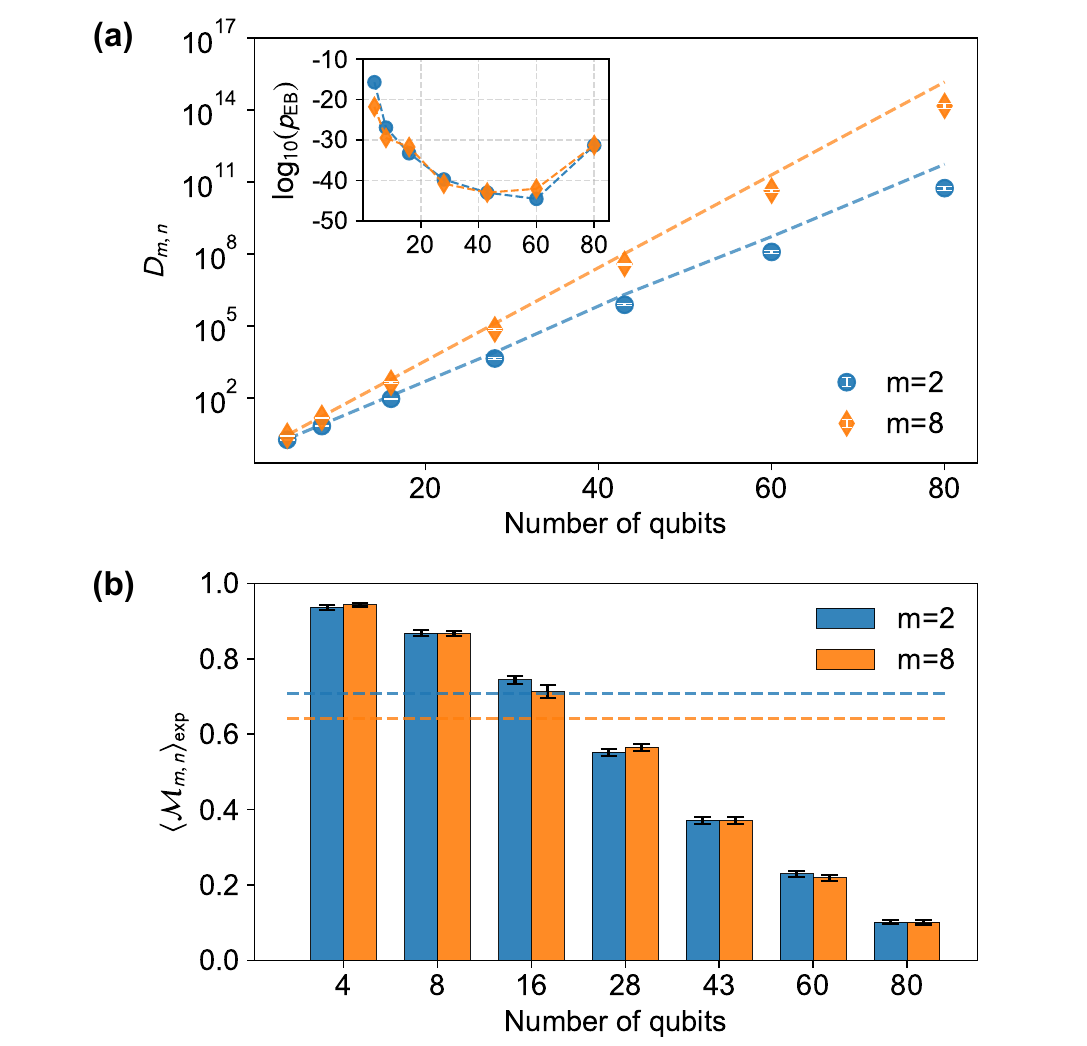}
    \caption{\textbf{Large-scale Bell violation and scaling performance.} 
    (a) 
    Experimental Bell ratios $D^{\mathrm{exp}}_{m,n}$ for the standard Mermin test ($m=2$, blue circles) and the generalized Mermin test ($m=8$, orange diamonds) up to $n=80$. Dashed lines show the corresponding ideal scalings $D^{\mathrm{id}}_{m,n}=1/C_{m,n}$. Error bars denote $40\times$ the standard error and are enlarged for visibility because the unscaled errors are smaller than the symbols. 
    The inset shows setting-level empirical-Bernstein upper bounds against the LHV model as $\log_{10}(p_{\mathrm{EB}})$, with lower values indicating stronger rejection. Color coding follows the main panel.
    (b) Measured expectation values $\langle\MM_{m,n}\rangle_{\mathrm{exp}}$; colored dashed lines mark the biseparable bounds. Error bars denote $10\times$ the standard error.
    }
    \label{fig:size_scaling}
\end{figure}

This separation between measured correlations and classical threshold is what makes the generalized family experimentally useful. For the largest system, $n=80$, the measured Bell-operator expectation values remain near $0.1$ for both setting families and 
the settings-level EB test still gives small $p_{\mathrm{EB}}$ upper bounds for both $m=2$ and $m=8$ 
(Supplementary Tables S1 and S2). 
The enhanced $m=8$ Bell ratio demonstrates the generalized test's certification advantage: tighter classical bounds yield a stronger ratio from essentially unchanged correlations, rather than from improved state preparation.

The advantage of the generalized family appears mainly in the lower classical thresholds and the resulting improvements in Bell ratio and depth certification, rather than in the local-rejection $p$-value scale at fixed $N$. 
At fixed $N$, the EB $p$-value upper bounds are controlled mainly by the observed settings-level mean and variance, which are similar for $m=2$ and $m=8$. The gain from the generalized family appears instead in the lower classical thresholds, and therefore in the Bell ratios and depth certification.

The direct comparison also reveals a clear improvement in scaling. Using the measured GHZ data, we obtain the following experimental Bell-ratio scaling:
\begin{equation}
    D^{\mathrm{exp}}_{8,n}
    =
    \frac{\langle \MM_{8,n}\rangle_{\mathrm{exp}}}{C_{8,n}}
    \propto
    (1.5168\pm0.0024)^n,
\end{equation}
which substantially exceeds the standard two-setting scaling. Increasing the local measurement complexity therefore produces a larger exponential Bell-ratio scaling factor for the same noisy GHZ state.

A cleaner isolation of the generalized Mermin advantage is obtained by fixing the state size and varying only the Bell-inequality choice. For the $80$-qubit GHZ state, \cref{fig:depth_scaling}(a) shows that the experimental Bell ratio $D^{\mathrm{exp}}_{m,80}$ increases monotonically as $m$ is raised from $2$ to $4,8,16$, and $32$. Once $m>2$, the observed Bell ratio already exceeds the theoretical ceiling of the standard Mermin inequality. At the same time, the expectation value changes only weakly with $m$. This behavior is consistent with the mechanism of the construction: larger $m$ yields a larger quantum-to-classical separation for essentially the same experimental correlator.

At the many-body level, the clearest consequence of the generalized family appears in the depth hierarchy. Applying the grouping-model bound \cref{eq:depth_main} to the measured data yields the nonlocality depths shown in \cref{fig:depth_scaling}(b). The generalized test outperforms the standard Mermin test once decoherence becomes appreciable: at $n=60$ we certify depth $12$ for $m=8$ versus depth $10$ for $m=2$, and at $n=80$ the advantage grows to depth $14$ versus depth $10$. Thus the generalized family rules out substantially larger classical groupings on the same physical GHZ states. Statistical confidence for each depth claim is assigned by testing the corresponding $(k-1)$-producible null hypothesis using the one-sided empirical Bernstein bound at the settings level; the associated $p_{\mathrm{EB}}$ values are summarized in Table S4 of the Supplementary Material.
{The non-monotonic depth trend arises from a staircase competition between discrete grouping-model thresholds and decoherence-driven correlator decay, yielding temporary increases at intermediate $n$ and plateaus or declines at larger $n$.}

\begin{figure}[t]
    \centering
    \includegraphics[width=\columnwidth]{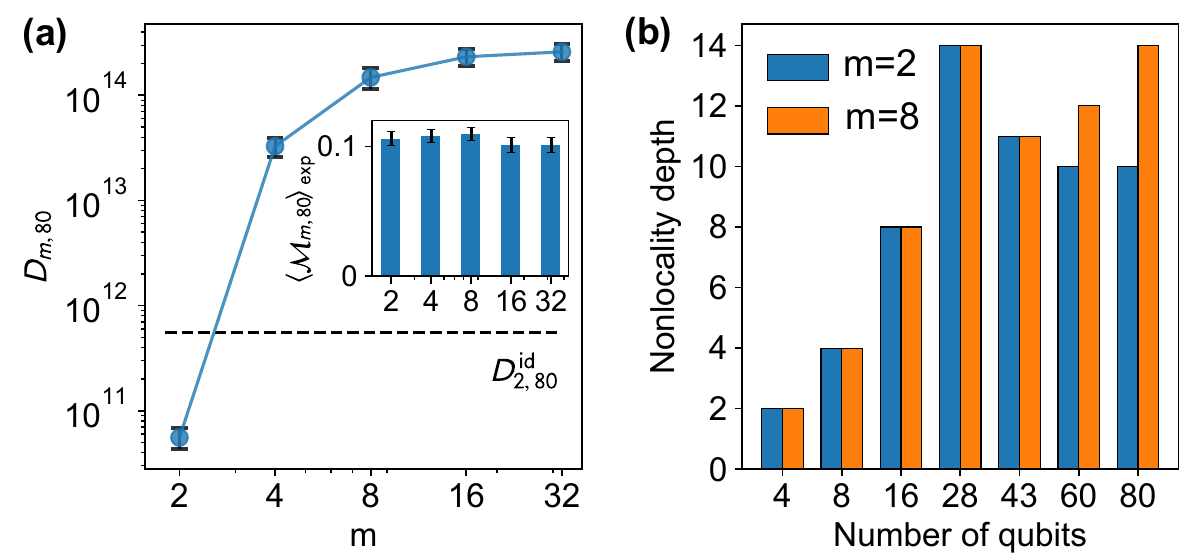}
    \caption{\textbf{Generalized Mermin advantage at fixed size and nonlocality depth certification.}
    (a) Experimental Bell ratio for the $80$-qubit GHZ state versus the number of settings $m$. The horizontal dashed line marks the ideal standard-Mermin value $D^{\mathrm{id}}_{2,80}$. The inset shows the corresponding expectation values. Error bars denote $40\times$ and $10\times$ the standard error in the main panel and inset, respectively.
    (b) Certified nonlocality depth versus system size for $m=2$ and $m=8$. 
    }
    \label{fig:depth_scaling}
\end{figure}

At smaller scale we also compare against the biseparable quantum benchmark \cref{eq:bisep_main}. For $n=16$ we obtain
$
    \langle \MM_{2,16}\rangle_{\mathrm{exp}}=0.7446
    >
    {1}/{\sqrt{2}}$
for the standard Mermin test and
$
    \langle \MM_{8,16}\rangle_{\mathrm{exp}}=0.7128
    >
    \frac{1}{8\sin(\pi/16)}\approx0.6407$
for the generalized test. 
The generalized setting almost doubles the gap to the relevant biseparable threshold. We therefore treat the GME comparison as a secondary but informative benchmark. It is not intended as a loophole-free or fully device-independent GME claim. The same generalized Bell design that improves nonlocality depth also makes the biseparable comparison more favorable at moderate size.

\par\bigskip

\noindent\textit{Discussion.}
{
We have shown that, in the noisy macroscopic regime, the choice of Bell functional is central to the certification strength. The standard Mermin inequality establishes the canonical exponential Bell scaling for GHZ states, 
while generalized Mermin inequalities show that increasing the measurement-setting number $m$ strengthens the certification. 
For the same noisy GHZ state and comparable measured Bell expectation values, larger $m$ primarily improves certification by tightening the relevant classical bounds. This yields stronger Bell ratios, sharper fixed-size tests, and deeper nonlocality-depth certification.
}


Experimentally, we prepared GHZ states of up to $80$ qubits on the \textit{Zuchongzhi 3.1} processor~\cite{jiang2026one} and observed Bell-violation ratios that continue to grow exponentially with system size in this large-$n$ regime. Compared with earlier GHZ-state Bell experiments, our results substantially extend the scale of direct GHZ-state Bell certification on superconducting hardware. Just as importantly, the reported Bell estimates and confidence statements rely only on directly measured correlators and analytical bounds, with no readout correction, tomography, or model-based mitigation~\cite{bonsel2025generating,vsupic2020self}. 

Our results suggest a simple design principle for macroscopic Bell tests in the NISQ era: for a fixed, noisy many-body state, certification strength depends on both the Bell functional and the state size. 
For GHZ states, higher-setting generalized Mermin tests make this concrete by strengthening certification on noisy hardware.

{
Finally, the present superconducting experiment is not loophole-free; a brief discussion of the relevant loopholes is provided in the Supplementary Material. Accordingly, we do not present it as a fully device-independent Bell test on chip. The experiment is therefore best interpreted as a correlation-only Bell benchmark based on observed correlators and analytical bounds. Generalized Mermin inequalities thus provide a sharper route to large-scale Bell certification, and they remain compatible with future Bell architectures that may also close the standard loopholes.
}

\begin{acknowledgments}

\textbf{Funding:}
This research was supported by the Quantum Science and Technology-National Science and Technology Major Project (Grant No.~2023ZD0300200 and No.~2021ZD0300200), Anhui Initiative in Quantum Information Technologies, the Special funds from Jinan Science and Technology Bureau and Jinan high tech Zone Management Committee, Deutsche Forschungsgemeinschaft (DFG, German Research Foundation, project number 563437167), the Sino-German Center for Research Promotion (Project M-0294), and the German Federal Ministry of Research, Technology and Space (Project QuKuK, Grant 
No.\ 16KIS1618K and Project BeRyQC, Grant No.\ 13N17292).

\end{acknowledgments}

\bibliographystyle{apsrev4-2}
%

\clearpage
\onecolumngrid

\makeatother
\let\author\revtexauthor
\let\affiliation\revtexaffiliation
\let\email\revtexemail
\let\thanks\revtexthanks
\setcounter{affil}{0}
\setcounter{mpfootnote}{0}
\makeatletter
\global\let\@FMN@list\@empty
\def\@mpfn{mpfootnote}
\makeatother
\title{Supplemental Material for\\
``Generalized Mermin Inequalities for Benchmarking Large-Scale GHZ States''}
\author{Jianbin Cai}
\thanks{These authors contributed equally to this work.}
\HFRC
\SHRC
\HFNL

\author{Junxiang Huang}
\thanks{These authors contributed equally to this work.}
\PKUCFCS
\PKUSCS

\author{Fynn Otto}
\thanks{These authors contributed equally to this work.}
\SIEGEN

\author{Yuan Li}
\thanks{These authors contributed equally to this work.}
\HFRC
\SHRC
\HFNL

\author{Carlos de Gois}
\SIEGEN
\INRIA

\author{Tao Jiang}
\HFRC
\SHRC
\HFNL

\author{Sirui Cao}
\HFRC
\SHRC
\HFNL

\author{Fangzheng Chen}
\LINKELAB

\author{Hao Fu}
\LINKELAB

\author{Jin Lin}
\HFRC
\SHRC
\HFNL

\author{Wei Xie}
\LINKELAB

\author{Naibin Zhou}
\HFRC
\SHRC
\HFNL

\author{Shibiao Tang}
\QCTek

\author{Xiang-Yang Li}
\HFNL
\LINKELAB

\author{Cheng-Zhi Peng}
\HFRC
\SHRC
\HFNL

\author{Xiao Yuan}
\PKUCFCS
\PKUSCS

\author{Otfried G\"uhne}
\email{otfried.guehne@uni-siegen.de}
\SIEGEN

\author{Ming Gong}
\email{minggong@ustc.edu.cn}
\HFRC
\SHRC
\HFNL

\revtexmaketitle
\onecolumngrid

\setcounter{section}{0}
\setcounter{subsection}{0}
\setcounter{subsubsection}{0}
\setcounter{equation}{0}
\setcounter{figure}{0}
\setcounter{table}{0}
\setcounter{secnumdepth}{3}
\renewcommand{\thesection}{S\Roman{section}}
\renewcommand{\theequation}{S\arabic{equation}}
\renewcommand{\thefigure}{S\arabic{figure}}
\renewcommand{\thetable}{S\arabic{table}}
\renewcommand{\theHsection}{supp.section.\arabic{section}}
\renewcommand{\theHequation}{supp.equation.\arabic{equation}}
\renewcommand{\theHfigure}{supp.figure.\arabic{figure}}
\renewcommand{\theHtable}{supp.table.\arabic{table}}
\setlength{\parskip}{0pt}
\setlength{\parindent}{10pt}
\renewcommand{\baselinestretch}{1.1}\normalsize
\vspace{1em}
\noindent\rule{\textwidth}{0.4pt}

\tableofcontents

\vspace{1em}
\noindent\rule{\textwidth}{0.4pt}

\renewcommand{\thefigure}{S\arabic{figure}}
\renewcommand{\thetable}{S\arabic{table}}

\section{Experimental Platform and GHZ Diagnostics}
\label{sec:si_platform}

\subsection{Processor Characterization}

    The experiments in this work were performed on the \textit{Zuchongzhi 3.1} superconducting quantum processor~\cite{sm:jiang2026one}. Its architecture features an array of 105 functional transmon qubits interconnected by 185 tunable couplers and 15 readout lines. More detailed design and comprehensive characterization were previously reported in Ref.~\cite{sm:jiang2026one}. In Fig.~\ref{fig:gr_errors}, we present the representative error rates for the device used in this experiment, including single-qubit (SQ) and two-qubit (CZ) gate errors characterized by cross-entropy benchmarking (XEB), as well as readout assignment errors ($1-f_{00}$ and $1-f_{11}$), where $f_{00}$ ($f_{11}$) denotes the probability of correctly assigning the qubit state $\ket{0}$ ($\ket{1}$) during readout. 

    \begin{figure}[h]
        \begin{center}
        \includegraphics[width=0.55\textwidth]{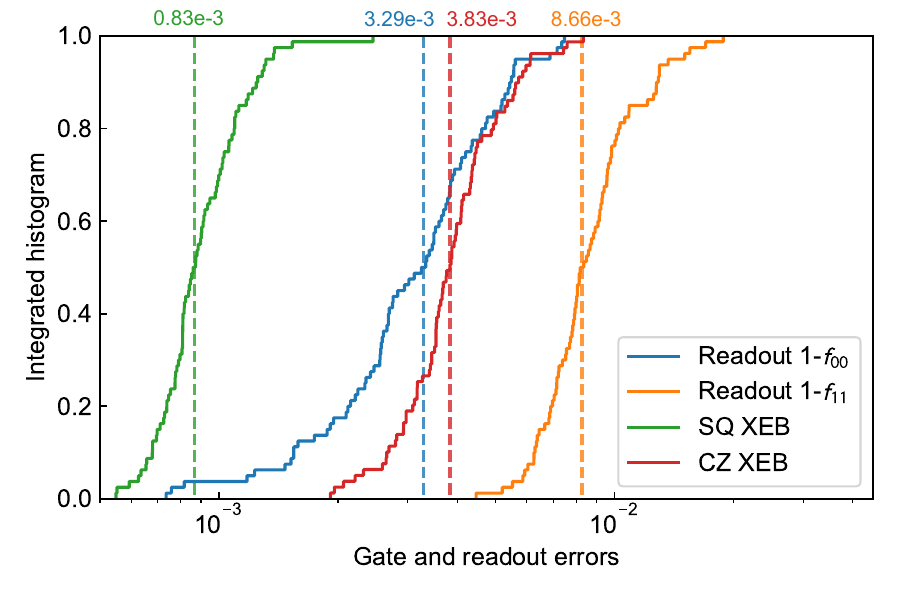}
        \end{center}
        \caption{Representative gate and readout characterization of the processor.
        The cumulative distributions summarize single-qubit gate errors, two-qubit CZ-gate errors, and readout assignment errors.
        These calibration data are reported only as hardware context; they do not enter any Bell estimate or statistical hypothesis test. 
        }
        \label{fig:gr_errors}
    \end{figure}

    \subsection{GHZ State Benchmarking}
        To characterize the quality of the prepared $n$-qubit Greenberger-Horne-Zeilinger (GHZ) states, we employ a randomized stabilizer sampling protocol to visualize how different noise channels affect GHZ correlators. The density matrix of an ideal GHZ state can be expressed via its stabilizer group expansion,
        \begin{equation}
            \rho_{GHZ} = \frac{1}{2^n} \sum_{S \in \mathcal{S}} S = \frac{1}{2^n} \prod_{k \in V} (I + S_k),
        \end{equation}
        where $V$ represents the set of all vertices in the graph. The stabilizer generators $\{S_k\}$ for an $n$-qubit GHZ state are defined as:
        \begin{equation}
            S_1 = \bigotimes_{j=1}^n X_j, \quad S_k = Z_1 Z_k \text{ for } k \in \{2, \dots, n\}.
        \end{equation}
        Here qubit $1$ remains the root qubit in the star-graph convention for the $Z$-type generators.
        
        Due to the exponential growth of the stabilizer group ($2^n$), an exhaustive measurement of every expectation value is computationally intractable for large $n$. We instead utilize a randomized fidelity estimation method~\cite{sm:flammia2011direct}: for each circuit iteration, we randomly sample a Pauli string $P = \prod_{k=1}^n P_k$, where each $P_k$ is drawn from $\{I, S_k\}$, and evaluate its expectation value. These stabilizers fall into two categories according to the setting of the first qubit, $P_1$:
        \begin{itemize}
            \item \textbf{Population ($Z$-type) Stabilizers ($P_1 = I$):} These terms involve only $Z$ and $I$ operators, corresponding to the populations (diagonal elements) of the density matrix.
            \item \textbf{Coherence ($X/Y$-type) Stabilizers ($P_1 = S_1$):} These involve measurements in the $XY$-plane, corresponding to the multi-qubit coherences (off-diagonal elements), and these measurements are also
            central in the standard Mermin inequality.
        \end{itemize}
        We define the total state fidelity of the $n$-qubit GHZ state as the arithmetic mean of the population and coherence contributions. This split is diagnostically useful: the former tracks diagonal populations (resilience to bit-flips), while the latter probes the fragile off-diagonal coherences essential for multi-body entanglement. 
        
       Experimental results for system sizes ranging from $n=2$ to $n=80$ show that the population term consistently outperforms the coherence term. As illustrated in Fig.~\ref{fig:ghz_fidelity}(a), both components decay with increasing qubit number, but the coherence term (orange dashed line) exhibits a more pronounced decline. For the $n=80$ case, the distribution of 4,600 sampled stabilizer measurements resolves into two clearly separated Gaussian peaks, as shown in Fig.~\ref{fig:ghz_fidelity}(b). Each measurement setting is repeated $5000$ times to ensure sufficient statistics. The diagonal population contribution yields a mean value of $\mu = 0.3646$ with a standard deviation of $\sigma = 0.0384$, while the off-diagonal coherence peak is centered at a significantly lower value of $\mu = 0.1012$ with $\sigma = 0.0335$. 
       
       The separation between the two components reflects their different sensitivities to experimental noise. Population stabilizers probe computational-basis populations and are primarily affected by population errors and state-assignment errors. Coherence stabilizers probe the global phase coherence of the GHZ superposition and are additionally sensitive to accumulated dephasing and phase drift across the full register. Although additional basis rotations are required to measure the coherence stabilizers, the principal physical distinction is their sensitivity to fragile many-body coherence. 
       As the system size $n$ increases, the cumulative effect of these phase fluctuations leads to an exponential decay in coherence and a widening gap between the two, as illustrated in Fig.~\ref{fig:ghz_fidelity}(a).
       
        \begin{figure}[h]
            \centering
            \includegraphics[width=\textwidth]{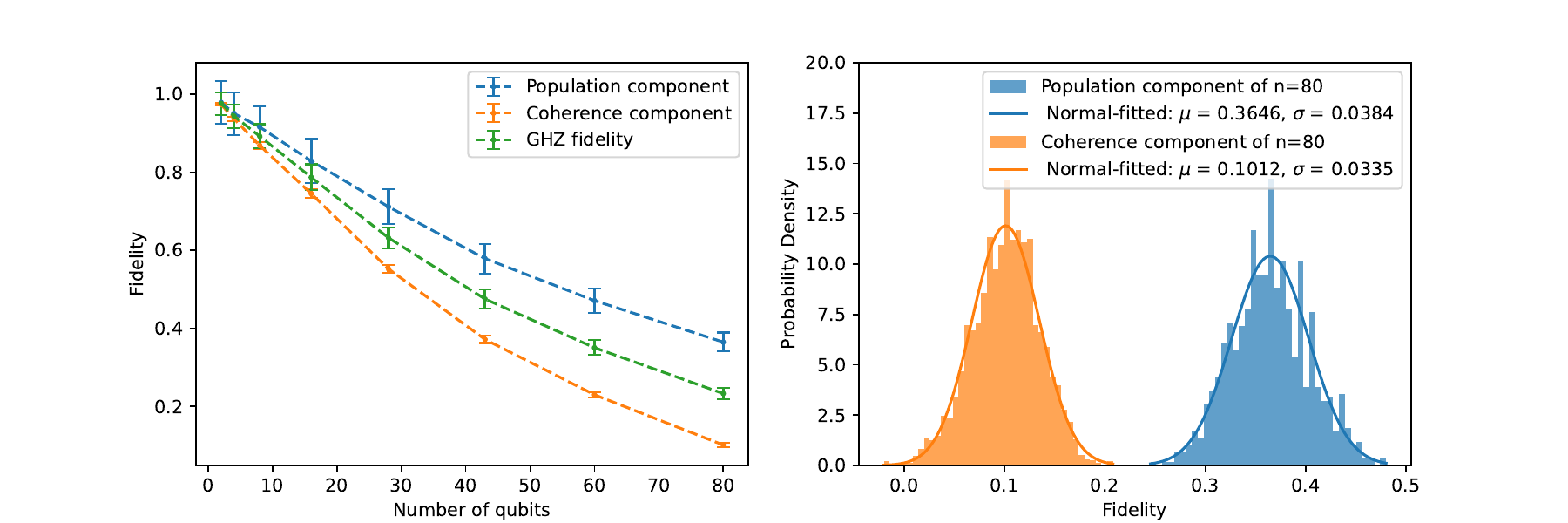}
            \caption{GHZ state fidelity characterization and statistical distribution. 
            (a) Measured population component (blue), coherence component (orange), and resulting GHZ-state fidelity (green) as functions of the system size $n$. The GHZ-state fidelity is calculated as the arithmetic mean of the population and coherence components. Error bars denote $10\times$ standard error.
            (b) Probability-density histogram of 4,600 randomized stabilizer measurements for the $n=80$ GHZ state. The data clearly resolve into a bimodal Gaussian structure, with the right peak representing diagonal populations ($\mu = 0.3646$) and the left peak representing off-diagonal coherences ($\mu = 0.1012$). Solid lines indicate Gaussian fits to the experimental data. 
            }
            \label{fig:ghz_fidelity}
        \end{figure}
    
       All values reported in this section are unmitigated stabilizer-based fidelity estimates obtained directly from the measured correlators, without readout-error correction or other model-based error mitigation. They therefore characterize the end-to-end prepared-and-measured GHZ correlations and include contributions from both state-preparation and measurement errors. Readout-error mitigation could yield larger inferred fidelities, but such estimates would depend on an additional calibration model and are not used here.

    \subsection{Spatial noise correlations and crosstalk}
        To assess potential classical explanations of the observed Bell violations (e.g., spatially correlated errors and measurement crosstalk),
        we empirically quantify pairwise correlations of flip-error events across the processor during both state-preparation and readout.
        These diagnostics do not enter the Bell estimates or statistical tests,
        but they provide supporting evidence that the observed multi-body Bell correlations are unlikely to be artifacts of classical correlated noise.
        We quantify these correlations by evaluating the covariance of flip error events $E_i$ and $E_j$ for any qubit pair ($q_i, q_j$)~\cite{sm:cao2023generation}, defined as
        \begin{equation}
            {\mathbf{cov}}[E_i, E_j] = \mathbb{E}[E_i \wedge E_j] - \mathbb{E}[E_i]\mathbb{E}[E_j], 
            \label{eq:appendix_cov}
        \end{equation}
        where $\mathbb{E}[E_i \wedge E_j]$ represents the joint probability of simultaneous flip errors. Using a randomized preparation and measurement scheme across the 80-qubit array, we characterized the statistical distribution of joint errors and covariances to generate cumulative distribution functions (CDFs) for various error types. As illustrated in Fig.~\ref{fig:corralation_m}, the median of measured covariances for all qubit pairs is predominantly clustered around $10^{-6}$ or less, with the median covariance remaining at least one order of magnitude smaller than the joint measurement error. These data indicate that spatial correlations during both state preparation and readout are weak on our device. The suppressed pairwise covariances therefore support the view that the measured Bell correlators cannot be straightforwardly attributed to simple classical noise correlations distributed across the chip.
    
        \begin{figure}[h]
            \centering
            \includegraphics[width=\textwidth]{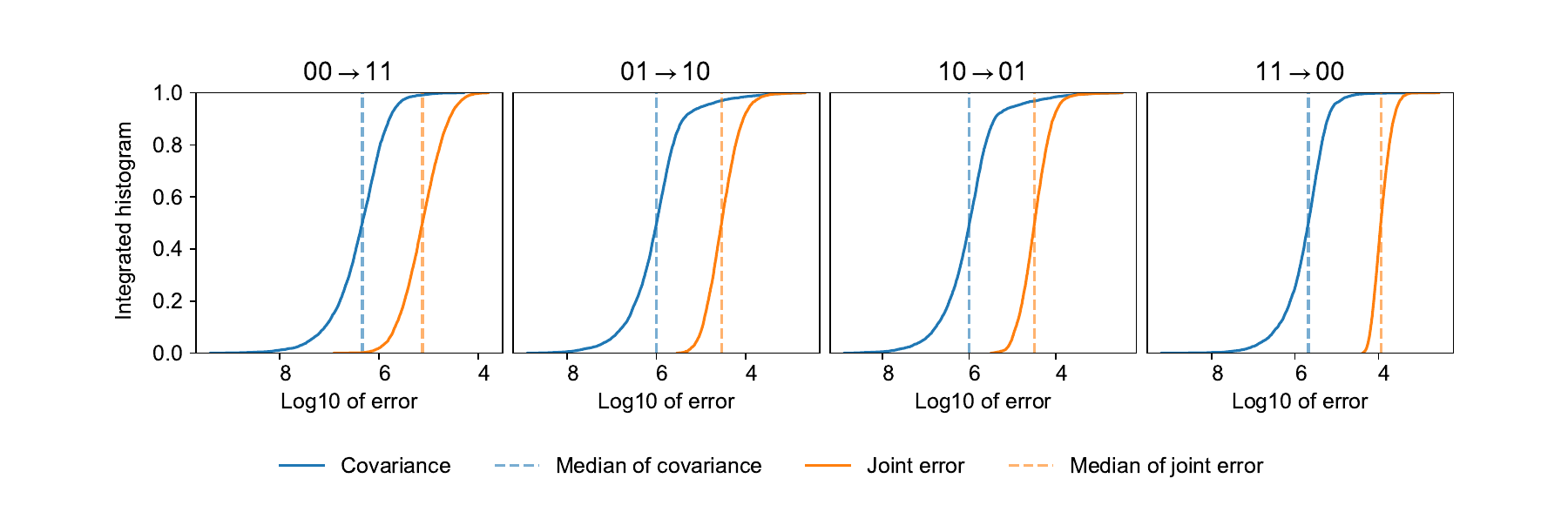}
            \caption{\textbf{Distributions of covariance and joint errors for the qubits forming the 80-qubit GHZ state.} 
            Cumulative distribution functions (CDFs) of the correlated noise, presented as integrated histograms. 
            The notation $jk \rightarrow \bar{j}\bar{k}$ denotes joint flip events, where an initial state $|j, k\rangle$ is measured as the bit-flipped state $|\bar{j}, \bar{k}\rangle$.
            }
            \label{fig:corralation_m}
        \end{figure}

\section{Settings-Level Empirical Bernstein Analysis}
\label{sec:si_statistics}

\subsection{Estimator}
For the normalized generalized Mermin operator
\begin{equation}
    \MM_{m,n}
    =
    \frac{1}{m^{n-1}}
    \sum_{\substack{\bm{x}\in\{0,\ldots,m-1\}^n\\ s(\bm{x})\equiv 0\ (\mathrm{mod}\ m)}}
    c_{\bm{x}} O_{\bm{x}},
\end{equation}
with
\begin{equation}
    c_{\bm{x}}=(-1)^{s(\bm{x})/m},
    \qquad
    O_{\bm{x}}=\bigotimes_{j=1}^{n}A^{(j)}_{x_j},
\end{equation}
we sample $N$ allowed settings $\bm{x}_i$ uniformly at random from the $m^{n-1}$ correlator terms and measure each sampled setting for $M$ shots.
If $y_{i,r}\in\{-1,+1\}$ denotes the $r$th shot outcome of the signed observable $c_{\bm{x}_i} O_{\bm{x}_i}$, we define the settings-level random variable
\begin{equation}
    S_i = \frac{1}{M}\sum_{r=1}^{M} y_{i,r}\in[-1,1].
\end{equation}
The estimator used throughout the paper is then
\begin{equation}
    \widehat{\MM}_{m,n}
    =
    \frac{1}{N}\sum_{i=1}^{N} S_i.
    \label{eq:si_estimator}
\end{equation}
The total number of physical shots is $K=NM$, but $K$ is used only as an experimental resource count.
Every statistical statement is formulated at the settings level with sample size $N$.
Because the sampled settings were randomized throughout the run, slow hardware drift is not systematically tied to any particular Bell term.
For the concentration analysis, we treat the settings-level observations $S_i$ as independent bounded variables, while allowing their distributions and expectations to vary across the run.
We therefore use the settings-level averages $S_i$ as the elementary bounded observations and take $N$, rather than the total shot count $K$, as the statistical sample size.
This is a conservative choice that does not rely on within-setting shot independence.

\subsection{Empirical Bernstein \texorpdfstring{$p$}{p}-value}
Let $C$ denote the relevant null-hypothesis threshold:
it can be the local bound, the biseparable bound, or a $k$-producible grouping-model bound.
Defining the empirical mean and empirical variance by
\begin{equation}
    \hat{\mu}=\frac{1}{N}\sum_{i=1}^{N} S_i,
    \qquad
    \hat{V}_N=\frac{1}{N-1}\sum_{i=1}^{N}(S_i-\hat{\mu})^2,
\end{equation}
and the population mean by
\begin{equation}
    \mu=\mathbb{E}[\hat{\mu}]=\frac{1}{N}\sum_{i=1}^{N}\mathbb{E}[S_i],
\end{equation}
we test the null hypothesis $H_0:\mu\le C$, where $\mu$ equals $\mathbb{E}[S_i]$ when the settings-level variables are identically distributed under uniform random sampling.
Theorem~11 of Ref.~\cite{sm:maurer2009empirical} yields the one-sided lower confidence bound
\begin{equation}
    \mu
    \ge
    \hat{\mu}
    -
    \sqrt{\frac{2\hat{V}_N\ln(2/\delta)}{N}}
    -
    \frac{14\ln(2/\delta)}{3(N-1)}
\end{equation}
with probability at least $1-\delta$.
Inverting this inequality gives the empirical-Bernstein $p$-value upper bound
\begin{equation}
    p_{\mathrm{EB}}
    =
    \inf\left\{\delta\in(0,1):
    \hat{\mu}
    -
    \sqrt{\frac{2\hat{V}_N\ln(2/\delta)}{N}}
    -
    \frac{14\ln(2/\delta)}{3(N-1)}
    \ge C
    \right\}.
    \label{eq:si_eb}
\end{equation}
Solving the boundary condition analytically yields the closed-form expression 
\begin{equation}
    p_{\mathrm{EB}} = 
    \min\left\{1,\,2\exp\left[
    -\left(\frac{3(N-1)}{28}\right)^2
    \left(
    \sqrt{
    \frac{2\hat{V}_N}{N}
    +
    \frac{56(\hat{\mu}-C)}{3(N-1)}
    }
    -
    \sqrt{\frac{2\hat{V}_N}{N}}
    \right)^2
    \right]\right\},
    \label{eq:si_eb_cf}
\end{equation}
for $\hat{\mu}>C$, and $p_{\mathrm{EB}}=1$ otherwise.
Since the $p$-values span multiple orders of magnitude in our experiments, we report the empirical-Bernstein upper bounds on a base-10 logarithmic scale, $\log_{10}(p_{\mathrm{EB}})$, 
for presentation.
No concentration inequality based on total-shot counts is used anywhere in the paper.

\subsection{Bell data summaries}
Tables~\ref{tab:ghz_m2} and~\ref{tab:ghz_m8} collect the settings-level EB analysis for the two Bell families emphasized in the main text: the standard Mermin test with $m=2$ and the generalized Mermin test with $m=8$.
The tables make three points explicit:
(i) the expectation values of the correlators decrease with system size but remain well resolved and significantly above zero even at $80$ qubits,
(ii) the local threshold falls exponentially faster, making large-scale Bell violation possible,
and (iii) the settings-level EB $p$-value upper bounds remain comfortably below conventional significance thresholds across the reported GHZ sizes.
Throughout these summaries, the reported Bell estimates and statistical tests use sampled GHZ correlators and analytical Bell-type bounds, without readout correction, state tomography, or a calibrated noise model. The empirical Bernstein analysis uses the settings-level sample size $N$, not the total number of physical shots $K=NM$.

\begin{table}[ht]
\centering
\begin{tabular}{lccccccc}
\toprule
$n$ & 80 & 60 & 43 & 28 & 16 & 8 & 4 \\
\midrule
$M$ & 1500 & 1500 & 1500 & 1500 & 1500 & 1500 & 1500 \\
$N$ & 3600 & 2200 & 1300 & 800 & 500 & 400 & 400 \\
$\langle \MM_{2,n}\rangle_{\exp}$ & 0.101 & 0.230 & 0.371 & 0.552 & 0.745 & 0.868 & 0.936 \\
$\log_{10}C_{2,n}$ & -11.740 & -8.730 & -6.322 & -3.913 & -2.107 & -0.903 & -0.301 \\
$\log_{10}(p_{\mathrm{EB}})$ & -31.329 & -44.587 & -43.046 & -39.781 & -33.306 & -26.944 & -15.711 \\
\bottomrule
\end{tabular}
\caption{Standard Mermin Bell data for GHZ states. 
The empirical-Bernstein $p$-value upper bounds are computed using the number of sampled settings $N$ as the sample size and reported as $\log_{10}(p_{\mathrm{EB}})$.
}
\label{tab:ghz_m2}
\end{table}

\begin{table}[ht]
\centering
\begin{tabular}{lccccccc}
\toprule
$n$ & 80 & 60 & 43 & 28 & 16 & 8 & 4 \\
\midrule
$M$ & 1500 & 1500 & 1500 & 1500 & 1500 & 1500 & 1500 \\
$N$ & 3600 & 2200 & 1300 & 800 & 500 & 400 & 400 \\
$\langle \MM_{8,n}\rangle_{\exp}$ & 0.101 & 0.219 & 0.371 & 0.565 & 0.713 & 0.867 & 0.944 \\
$\log_{10}C_{8,n}$ & -15.165 & -11.298 & -8.020 & -5.112 & -2.792 & -1.245 & -0.464 \\
$\log_{10}(p_{\mathrm{EB}})$ & -31.276 & -42.086 & -43.033 & -40.835 & -31.752 & -29.506 & -21.709 \\
\bottomrule
\end{tabular}
\caption{Generalized Mermin Bell data for GHZ states with $m=8$. The observed correlators are close to the $m=2$ case, while the lower classical threshold produces a stronger experimental Bell ratio.}
\label{tab:ghz_m8}
\end{table}

\section{Theory of Generalized Mermin Inequalities}
\label{sec:si_gmi_theory}

\vspace{1em}

\noindent
In this appendix we investigate the local hidden-variable bounds (also called local or classical bounds), and the nonlocality-depth bounds for the generalized Mermin expressions.
As we will show, a useful effect appears in the scaling of the classical bound when increasing the number of local measurement settings, leading to a significant improvement on noise robustness when compared to the standard Mermin inequality.

\subsection{The generalized Mermin operator and statement of the main results}\label{sec:intro_mermin}

\noindent
In the main text we introduced the Bell operator
\begin{equation}
    \MM_{m,n}
    =
    \frac1{m^{n-1}}\sum_{\substack{\bm{x}\in\{0,\ldots,m-1\}^n\\ s(\bm{x})\equiv 0 \,(\mathrm{mod}\, m)}}
    (-1)^{s(\bm{x})/m}
    \bigotimes_{j=1}^n A_{x_j}^{(j)},
    \label{eq:gen_mermin_op}
\end{equation}
where $\bm{x}=(x_1,\ldots,x_n)$ is the vector of local measurement settings and $s(\bm{x})=\sum_{j=1}^n x_j$.
The parameters $n$ and $m$ correspond, respectively, to the number of parties and the number of measurement settings per party.
This expression can be seen as a generalization of the Mermin inequality \cite{sm:mermin1990extreme} to any number of measurement settings, and reduces to it when $m = 2$.

\smallskip
As an example, for $m=4$ and $n = 3$, 
the $\MM_{4, 3}$ operator contains the $4^2$ triples with $x_1+x_2+x_3\in\{0,4,8\}$:
\begin{equation}
\begin{aligned}
    \MM_{4,3}=\frac1{16}\Big(&A_0A_0A_0
    - \sum_{\mathrm{perm.}} A_0A_1A_3
    - \sum_{\mathrm{perm.}} A_0A_2A_2
    - \sum_{\mathrm{perm.}} A_1A_1A_2
    + \sum_{\mathrm{perm.}} A_2A_3A_3\Big),
\end{aligned}
\end{equation}
where each product is a tensor product over the three parties, the sums run over distinct permutations of the displayed factors, and party labels are inferred by the ordering.

\medskip
The form of \cref{eq:gen_mermin_op} is particularly interesting when applied to the GHZ state.
Indeed, using the equatorial measurements
\begin{equation}
    A_x = \cos\!\left(\frac{\pi x}{m}\right)X+\sin\!\left(\frac{\pi x}{m}\right)Y,
    \qquad x\in\{0,\ldots,m-1\},
\end{equation}
it follows that
\begin{equation}
    \bra{\mathrm{GHZ}_n}\bigotimes_{j=1}^n A_{x_j}^{(j)}\ket{\mathrm{GHZ}_n}
    =\cos\!\left(\frac{\pi s(\bm{x})}{m}\right).
\end{equation}
Consequently, whenever $s(\bm{x})\equiv0\pmod m$, the corresponding tensor product stabilizes $\GHZ{n}$ with eigenvalue $(-1)^{s(\bm{x})/m}$.
The summation in \cref{eq:gen_mermin_op} selects precisely these terms, and the coefficients are such that each term contributes a $+1$ on a $\GHZ{n}$ measured with the equatorial measurements.
Thus the algebraic maximum of the generalized Mermin expression is achievable in quantum theory, and increasing $m$ has the effect of adding more GHZ stabilizer correlators.

\medskip
The generalized Mermin expressions can also be interpreted as a family of pseudo-telepathy games \cite{sm:brassard2005quantum}.
    Such games are interesting because they turn Bell nonlocality into a communication-free cooperative task and give an operational form to the separation between quantum and classical resources~\cite{sm:brassard2005quantum}.
    In our case, the players receive inputs $x_j\in\{0,\ldots,m-1\}$ promised to satisfy $s(\bm{x})\equiv0\pmod m$, each player outputs a bit $y_j$, and they win when $\sum_{j=1}^n y_j \equiv \frac{s(\bm{x})}{m}\pmod 2$.
    For $m = 2$, this is Mermin's parity game, for $m=2^\ell$, this is the parity-game family with $\ell$-bit inputs, and with the parameter choice $\ell=\lceil\lg n\rceil-1$ they are known as extended parity games~\cite{sm:buhrman2003combinatorics,sm:brassard2005quantum}.
    To the best of our knowledge, the exact classical value of this family of games is an open problem.
    The local bounds derived below solve this problem in the regimes covered by \Cref{prop:bound_even,prop:bound_odd}.

\medskip
The main results of this appendix are stated in \Cref{prop:bound_even,prop:bound_odd} and \Cref{cor:nonlocal_depth}.
First, we argue that restricting $m$ to powers of two captures the cases relevant for improving the Bell ratio, because adding an odd factor to $m$ cannot improve the ratio between the GHZ quantum value and the local bound (\Cref{lem:powers-of-two}).
Then, \Cref{prop:bound_even} covers the classical bounds for all sufficiently large even numbers of parties (see conditions in the statement), yielding
\begin{equation}
    C_{m,n}=\frac 2{m^n}\sum_{j\in[\frac m2]}\sin^{-n}\!\left(\frac{(2j+1)\pi}{2m}\right).
\end{equation}
The case for odd $n$ is discussed in \Cref{prop:bound_odd}, where we find the bound
\begin{equation}
    C_{m,n}=\frac 2{m^n}\sum_{j\in[\frac m2]}(-1)^j\cos\!\left(\frac{(2j+1)\pi}{2m}\right)\sin^{-n}\!\left(\frac{(2j+1)\pi}{2m}\right).
\end{equation}
Lastly, \Cref{cor:nonlocal_depth} shows that the nonlocality depth bounds in terms of the maximal group size $p$, where the parties are allowed to communicate inside each group, are given by
\begin{equation}
    \max_{p\text{-prod.}} \expval{\MM_{m,n}} = C_{m,\lceil n/p\rceil}.
\end{equation}

\subsection{An equivalent expression for the generalized Mermin operator}

For the derivation of the local bounds, it is convenient to rewrite the generalized Mermin operator in a product form.
For each party $j$ and each $k\in\{0,\ldots,m-1\}$, define
\begin{equation}\label{eq:phased_obs_sum}
    B_k^{(j)}=\sum_{l=0}^{m-1}e^{i\pi(2k+1)l/m}A_l^{(j)}.
\end{equation}
We will show that
\begin{equation}
    \MM_{m,n}=\frac1{m^n}\sum_{k=0}^{m-1}\bigotimes_{j=1}^n B_k^{(j)}.
    \label{eq:inequality-short}
\end{equation}

For any fixed $k$, expanding the above tensor product gives
\begin{equation}
    \bigotimes_{j=1}^n B_k^{(j)}
    =
    \sum_{\bm{x}\in\{0,\ldots,m-1\}^n}
    e^{i\pi(2k+1)s(\bm{x})/m}\bigotimes_{j=1}^n A_{x_j}^{(j)}.
\end{equation}
Indeed, each choice of a setting vector $\bm{x}=(x_1,\ldots,x_n)$ picks the term $A_{x_j}^{(j)}$ from the $j$th factor and contributes the phase $e^{i\pi(2k+1)x_j/m}$.
Multiplying these local phases produces the collective phase $e^{i\pi(2k+1)s(\bm{x})/m}$.

\medskip
Substituting this expansion into the right-hand side of \cref{eq:inequality-short}, we find
\begin{equation}
    \frac1{m^n}\sum_{k=0}^{m-1}\bigotimes_{j=1}^n B_k^{(j)}
    =
    \sum_{\bm{x}\in\{0,\ldots,m-1\}^n} c(\bm{x})\bigotimes_{j=1}^n A_{x_j}^{(j)},
\end{equation}
with the coefficient
\begin{equation}
\begin{aligned}
    c(\bm{x})
    &=\frac1{m^n}\sum_{k=0}^{m-1}e^{i\pi(2k+1)s(\bm{x})/m}\\
    &=\frac{e^{i\pi s(\bm{x})/m}}{m^n}\sum_{k=0}^{m-1}e^{k \, (2\pi i s(\bm{x})/m)}.
\end{aligned}
\end{equation}
The remaining sum is the sum of all powers of the $m$th roots of unity $e^{2\pi i s(\bm{x})/m}$.
Hence
\begin{equation}
\sum_{k=0}^{m-1}e^{k \, (2\pi i s(\bm{x})/m)} =
    \begin{cases}
        m, & s(\bm{x})\equiv 0\pmod m, \\
        0, & s(\bm{x})\not\equiv 0\pmod m.
    \end{cases}
\end{equation}

This yields
\begin{equation}
    c(\bm{x})=
    \begin{cases}
        \dfrac1{m^{n-1}}(-1)^{s(\bm{x})/m}, & s(\bm{x})\equiv0\pmod m,\\[1ex]
        0, & s(\bm{x})\not\equiv0\pmod m.
    \end{cases}
\end{equation}
These are exactly the coefficients appearing in \cref{eq:gen_mermin_op}, proving \cref{eq:inequality-short}.

\medskip
When $m$ is even, the sum over $k$ can be reduced by pairing Hermitian conjugate terms. Indeed, for every integer $l$, 
\begin{equation}
    e^{i\pi(2(m-1-k)+1)l/m}=e^{-i\pi(2k+1)l/m},
\end{equation}
and hence $B_{m-1-k}^{(j)}=(B_k^{(j)})^\dagger$, implying that
\begin{equation}
\bigotimes_{j=1}^n B_{m-1-k}^{(j)}
=
\left(\bigotimes_{j=1}^n B_k^{(j)}\right)^\dagger .
\end{equation}
Pairing the terms with indices $k$ and $m-1-k$ then gives
\begin{equation}
\begin{aligned}
    \MM_{m,n}=\frac2{m^n}\operatorname{Re}\left[\sum_{k=0}^{\frac m2-1}\bigotimes_{j=1}^n B_k^{(j)}\right].
    \label{eq:inequality-short2}
\end{aligned}
\end{equation}

\subsection{Local bounds}

{We now determine the local bounds of the generalized Mermin expressions.
The proof is organized as follows.
First, \Cref{lem:powers-of-two} shows that multiplying the number of settings $m$ by an odd factor cannot improve the ratio between the maximum quantum value and the local bound, so it is sufficient to focus on $m = 2^x$.
We then use the product form from \cref{eq:inequality-short} and express the deterministic local strategies \cite{sm:brunner2014bell} in a Fourier representation [\cref{eq:vmk}].
This reduces the problem to a single-party upper bound in terms of the $n$-norm of an $m$-component vector [\cref{eq:norm_bound}].
For an even number of parties, this upper bound is tight because a deterministic strategy can be paired with its conjugate strategy (\Cref{lem:conjugate_strategy} and \Cref{cor:lb_norm}).
It remains to find which single-party strategies maximize this $n$-norm.
\Cref{lem:2norm,lem:max_norm} give the two ingredients for this step: the $2$-norm is fixed for all strategies, and the largest possible $\infty$-norm is attained by sign choices that align equally spaced unit vectors in one half-plane.
For sufficiently large even $n$, norm interpolation shows that the maximizers of the $\infty$-norm also maximize the relevant $n$-norm, yielding the local bounds in \Cref{prop:bound_even}.
The case for odd $n$ is then analyzed in \Cref{prop:bound_odd}.
}

\medskip
\medskip
First, we show that the interesting cases are those for which $m=2^x$, because taking $m=2^x y$ with $y$ being an odd integer does not improve the ratio between the local and quantum bounds.
\begin{lem}
    The local bound $C_{2^x y,n}$ is lower-bounded for odd $y$ as
    \begin{equation}
        C_{2^x y,n}\geq C_{2^x,n}
    \end{equation}
    while the quantum bound does not change: $ Q_{2^x y,n}= Q_{2^x,n}=1 $. Hence, the maximal quantum violation does not increase:
    \begin{equation}
        \frac{Q_{2^xy,n}}{C_{2^xy,n}}\leq \frac{Q_{2^x,n}}{C_{2^x,n}}.
    \end{equation}
\label{lem:powers-of-two}
\end{lem}
\begin{proof}
    It is sufficient to consider deterministic strategies (with $\pm 1$ assignments) to investigate the local bounds. Let $\bm a^{(j)}\in\{\pm1\}^{2^xy}$ and $\bm b^{(j)}\in\{\pm1\}^{2^x}$ be the strategies of party $j$ for the case of $m=2^xy$ and $m=2^x$ measurement settings, respectively. We show that the expectation values achieved with the local strategies $\bm b^{(j)}$ for $m=2^x$ can be attained for $m=2^xy$ with a suitable choice of local strategies $\bm a^{(j)}$, too. We define the expectation values for the operators $B^{(j)}_k$ (cf.\ \cref{eq:phased_obs_sum}) given the strategy $\bm c^{(j)}$ as
    \begin{equation}
        v^{k}_{m}(\bm c^{(j)})= \sum_{l=0}^{m-1}  c^{(j)}_{l} e^{i\pi (2k+1)\frac{l}{m}}
    \label{eq:vmk}
    \end{equation}
    and the expectation value for the whole expression becomes
    \begin{equation}\label{eq:M_vmk}
    	\expval{\MM_{m,n}}_{\bm a^{(1)}, \ldots, \bm a^{(n)}}=\frac1{m^n}\sum_{k=0}^{m-1}\prod_{j=1}^n v_m^{k}(\bm a^{(j)}).
    \end{equation}
    Now, we extend the strategies $\bm b^{(j)}$ to $m=2^xy$ with $a^{(j)}_{2^x p+q}=(-1)^p b^{(j)}_q$ for all $p\in\{0,\dots ,y-1\}$ and $q\in\{0,\dots,2^x-1\}$ where $\bm b^{(j)}$ are arbitrary strategies for the case of $m=2^x$. We split the index of $\bm a^{(j)}$ into $p$ and $q$ such that we can explicitly compute the sum over $p$ as a geometric sum: 
    \begin{equation}
    \begin{aligned}
        v_{2^x y}^{k}(\bm a^{(j)})&=\sum_{q=0}^{2^x-1}\sum_{p=0}^{y-1}e^{i\pi (2k+1)\frac{2^x p+q}{2^x y}} (-1)^p b_q^{(j)}\\
        &=\sum_{q=0}^{2^x-1}e^{i\pi (2k+1)\frac{q}{2^x y}} b_q^{(j)}\underbrace{\sum_{p=0}^{y-1}\left(e^{i\pi \frac{2k+1}{y}(y+1)}\right)^p}_{=y\delta_{[2k+1]_y,0}}\\
        &=y\delta_{[2k+1]_y,0}\sum_{q=0}^{2^x-1}e^{i\pi \frac{(2k+1)}{y}\frac{q}{2^x}} b_q^{(j)}\\
        &=y\delta_{[2k+1]_y,0} v_{2^x}^{\frac{1}{2} \left( \frac{2k+1}{y} - 1 \right)}(\bm b^{(j)}).
    \end{aligned}
    \label{eq:partial-sum-aj-to-bj}
    \end{equation}
    To justify the Kronecker delta in \cref{eq:partial-sum-aj-to-bj}, let
    \[
        \omega=e^{i\pi \frac{2k+1}{y}(y+1)}.
    \]
    Since $y$ is odd, $\omega^y=1$, and because $\gcd((y+1)/2,y)=1$, we have $\omega=1$ if and only if $y$ divides $2k+1$. Therefore,
    \[
        \sum_{p=0}^{y-1}\omega^p
        =
        \begin{cases}
            y, & [2k+1]_y=0,\\
            0, & [2k+1]_y\neq 0,
        \end{cases}
        =
        y\delta_{[2k+1]_y,0}.
    \]
    This leads to
    \begin{equation}
        \begin{aligned}
            \expval{\MM_{2^x y,n}}_{\bm a^{(1)}, \ldots, \bm a^{(n)}}&=\frac1{(2^x y)^n} \sum_{k=0}^{2^x y-1}\prod_{j=1}^n v_{2^x y}^{k}(\bm a^{(j)})=\frac{1}{(2^x)^n}\sum_{k^\prime=0}^{2^x -1}\prod_{j=1}^n v_{2^x}^{k^\prime}(\bm b^{(j)})=\expval{\MM_{2^x,n}}_{\bm b^{(1)}, \ldots,\bm b^{(n)}},
        \end{aligned}
    \end{equation}
    where in the second equality we used \cref{eq:partial-sum-aj-to-bj} together with the fact that the condition $\delta_{[2k+1]_y,0}$ restricts the sum to those $k$ for which $2k+1=y(2k'+1)$, yielding the re-indexing $k'=\frac12\left(\frac{2k+1}{y}-1\right)\in\{0,\dots,2^x-1\}$.
    Taking $\bm b^{(j)}$ to be the optimal deterministic strategy in the $m = 2^x$ case shows that there is a strategy such that $C_{2^x y,n}\geq C_{2^x,n}$. As argued in \Cref{sec:intro_mermin}, the quantum bound is the same in both cases, thus concluding the proof.
\end{proof}

In what follows, we will determine the local bounds for the generalized Mermin expressions.
Due to the improvements in noise robustness occurring only when the number of measurements $m$ is a power of two, we will not consider other cases.

\medskip
For deterministic strategies $\bm a^{(1)},\dots,\bm a^{(n)}\in\{\pm1\}^m$, define
\begin{equation}
    \bm v_m(\bm a^{(j)})=\left( v_m^0(\bm a^{(j)}),\ldots,v_m^{m-1}(\bm a^{(j)}) \right) ,
\end{equation}
with $v_m^k(\bm a^{(j)})$ given as in \cref{eq:vmk}. With \cref{eq:M_vmk} the objective function can be bounded as
\begin{equation}
	m^n \expval{\MM_{m,n}}_{\bm a^{(1)}, \ldots, \bm a^{(n)}}=\sum_{k=0}^{m-1}\prod_{j=1}^n v_m^{k}(\bm a^{(j)})
    \leq \sum_{k=0}^{m-1}\prod_{j=1}^n \abs{ v_m^k(\bm a^{(j)}) }
    \leq \prod_{j = 1}^n \left( \sum_{k = 0}^{m-1} \abs{ v_m^k(\bm a^{(j)}) }^n \right)^{1 / n}
    = \prod_{j = 1}^n \lVert \bm v_m(\bm a^{(j)}) \rVert_n  ,
\end{equation}
where $\abs{\cdot}$ is the element-wise absolute value and {the generalized} Hölder's inequality  has been used for the second inequality.
Thus, it is sufficient to choose the vectors $\bm v_m(\bm a^{(j)})$ with maximal $n$-norm:
\begin{equation}
    \max_{\bm a^{(1)}, \ldots, \bm a^{(n)}} \expval{\MM_{m,n}}_{\bm a^{(1)}, \ldots, \bm a^{(n)}}
    \leq \frac{1}{m^n} \max_{\bm a \in \{ \pm 1 \}^m} \lVert \bm v_m(\bm a) \rVert_n^n .
    \label{eq:norm_bound}
\end{equation}
Therefore, we can upper bound the local bound by the maximum $n$-norm of the vector
$\bm v_m(\bm a)$, being a function of the single-party deterministic strategy vector 
$\bm a$.


This upper bound is tight for an even number of parties as a consequence of the following lemma.
\begin{lem}\label{lem:conjugate_strategy}
    For every deterministic strategy $\bm a \in  \{ \pm 1 \}^m$, there is a conjugate strategy $\tilde{\bm a} \in \{ \pm 1 \}^m$ such that $\bm v_m(\bm a) = \bm v_m(\tilde{\bm a})^\ast$,
    where $^\ast$ denotes the element-wise complex conjugate.
\end{lem}
\begin{proof}
    The strategy $\tilde{\bm a}$ is constructed as follows. We set $\tilde a_0=a_0$, because $a_0e^{i(2k+1)\frac{0\pi}{m}}$ is real for every $k$. The other elements $a_l$ with $l\in\{1,\ldots, m-1\}$ undergo the following change. If $a_l=a_{m-l}$, then
    \begin{equation}
    	a_l e^{i(2k+1)\frac{l\pi}{m}}+a_{m-l}e^{i(2k+1)\frac{(m-l)\pi}{m}}= a_l\left(e^{i(2k+1)\frac{l\pi}{m}}-e^{-i(2k+1)\frac{l\pi}{m}}\right)=2ia_l\sin((2k+1)\frac{l\pi}{m})
    \end{equation} 
    is imaginary and we choose $\tilde a_l=\tilde a_{m-l}=-a_l$. Otherwise, if $a_l=-a_{m-l}$, then
    \begin{equation}
    	a_l e^{i(2k+1)\frac{l\pi}{m}}+a_{m-l}e^{i(2k+1)\frac{(m-l)\pi}{m}}= a_l\left(e^{i(2k+1)\frac{l\pi}{m}}+e^{-i(2k+1)\frac{l\pi}{m}}\right)=2a_l\cos((2k+1)\frac{l\pi}{m})
    \end{equation} 
    is real and we therefore keep the original signs,
    \[
    \widetilde a_l=a_l,
    \qquad
    \widetilde a_{m-l}=a_{m-l}=-a_l.
    \]
    In the case of $m$ even, the first case applies for the self-paired index $l=m/2$:
    its contribution is purely imaginary and is conjugated by setting
    $\widetilde a_{m/2}=-a_{m/2}$.
    Plugging this into the definition of $\bm v_m(\bm a)$ leads to the claim.
\end{proof}

\begin{cor}\label{cor:lb_norm}
    The local bound for the case of an even number of parties can be expressed via vector norms as
    \begin{equation}
        C_{m,n}=\frac1{m^n} \max_{\bm a \in \{\pm 1 \}^m} \norm{\bm v_{m}(\bm a)}_n^n
    \end{equation}
    where $\bm a\in\{\pm1\}^m$ is a single-party deterministic strategy.
\end{cor}
\begin{proof}
	We choose a strategy $\bm a\in\{\pm1\}^m$  for all odd parties $\bm a^{(2j-1)}=\bm a$, $j\in\{1,\ldots \frac n2\}$ and its conjugate strategy for all even parties $\bm a^{(2j)}=\tilde{ \bm a}$, $j\in\{1,\ldots \frac n2\}$.  Then, $v^k_m(\bm a^{(2j-1)})v^k_m(\bm a^{(2j)})=\abs{ v_m^k(\bm a) }^2$ leads to
	\begin{equation}
		\expval{\MM_{m,n}}_{\bm a^{(1)}, \ldots, \bm a^{(n)}}=\frac1{m^n}\sum_{k=0}^{m-1}\prod_{j=1}^n v_m^{k}(\bm a^{(j)})=\frac1{m^n}\sum_{k=0}^{m-1}\prod_{j=1}^n \abs{v_m^{k}(\bm a)}=\frac1{m^n}\sum_{k=0}^{m-1} \abs{v_m^{k}(\bm a)}^n=\frac{1}{m^n}\norm{\bm v_m(\bm a)}_n^n.
	\end{equation}
	Maximizing over the strategies $\bm a$ saturates the upper bound in \cref{eq:norm_bound} and implies the claim.
\end{proof}

To proceed with the analysis, we will relate the $n$-norm to the $2$- and the $\infty$-norm via norm interpolation.
We start with some results regarding the $2$-norm and the $\infty$-norm.
\begin{lem}\label{lem:2norm}
    The $2$-norm of $\bm v_m(\bm a)$ does not depend on the single-party strategy $\bm a$:
    \begin{equation}
        \norm{\bm v_m(\bm a)}_2=m.
    \end{equation}
\end{lem}
\begin{proof}
    The inverse discrete Fourier transform of $\bm v_m(\bm a)$ is
    \begin{equation}
        (\mathcal{F}^{-1}\left[ \bm v_m(\bm a) \right])_l=m e^{i\pi \frac{l}{m}}a_l .
    \end{equation}
    Using Parseval's theorem,
    \begin{equation}
        \norm{\bm v_m(\bm a)}_2=\frac1{\sqrt m}\norm{\mathcal{F}^{-1}(\bm v_m(\bm a))}_2=m.
    \end{equation}
\end{proof}
\begin{lem}\label{lem:max_norm}
    The maximal value for the infinity norm of $\bm v_m(\bm a)$ is
    \begin{equation}
        \max_{\bm a}\norm{\bm v_m(\bm a)}_\infty=\frac{1}{\sin(\frac{\pi}{2m})}.
    \end{equation}
    Every strategy $\bm a$ with $\norm{\bm v_m(\bm a)}_\infty=\frac{1}{\sin{\frac{\pi}{2m}}}$ leads to the same $n$-norms with
    \begin{equation}
        \norm{\bm v_m(\bm a)}_n^n=\sum_{l=0}^{m-1}\frac{1}{\sin[n]((2l+1)\frac{\pi}{2m})}.
    \end{equation}
    If $\norm{\bm v_m(\bm b)}_\infty<\frac{1}{\sin{\frac{\pi}{2m}}}$ for another strategy $\bm b\in\{\pm1\}^m$, then $\norm{\bm v_m(\bm b)}_\infty\leq\frac{1}{\sin(\frac{\pi}{2m})}-4\sin(\frac{\pi}{2m})$.
\end{lem}
\begin{proof}
    We can visualize $\bm v^k_m(\bm a)$ as the sum of unit vectors in the complex plane (\cref{fig:unit_vec}) where the $a_l$ determine which element of an antipodal pair is chosen.
    Since $m = 2^x$, the phases $e^{i \pi (2k+1) \frac{l}{m}}$ for $l \in \{0, \ldots, m - 1\}$ are the same as in $e^{i \pi \frac{l}{m}}$ up to reordering and multiples of $\pi$.
    Thus, to maximize $\abs{v_m^k(\bm a)}$ over $\bm a \in \{\pm 1\}^m$ it is sufficient to maximize on the first component $v^{0}_m(\bm a)$.
    Its absolute value is maximized if all unit vectors lie in the same half-plane. Without loss of generality, assume that for the first component of $\bm v_m(\bm a)$, all unit vectors $ a_l e^{i\pi l/m} $ are in the same half-plane. This is only possible if the sign of $a_l$ changes at most once with increasing $l$ and it implies that, up to a global sign flip, every maximizing strategy is of the form
    \begin{equation}
    	a_l=-1 \text{ for } l<r,\qquad a_l=1 \text{ for } l\ge r,
    \end{equation}
    for some $ r\in \{0,\ldots, m-1 \}$. Then, the components of $\bm v_m(\bm a)$ are
        \begin{equation}\label{eq:max_comp}
    	\begin{aligned}
    		v^k_m(\bm a)&=\sum_{l=0}^{m-1}e^{i\pi (2k+1)\frac{l}{m}}a_l\\
    		&=\sum_{l=0}^{r-1}e^{i\pi (2k+1)\frac{l}{m}}e^{i\pi(2k+1)}+\sum_{l=r}^{m-1}e^{i\pi (2k+1)\frac{l}{m}}\\
    		&=\sum_{l=0}^{r-1}e^{i\pi (2k+1)\frac{l+m}{m}}+\sum_{l=r}^{m-1}e^{i\pi (2k+1)\frac{l}{m}}\\
    		&=\sum_{l=r}^{m-1+r}e^{i\pi (2k+1)\frac{l}{m}}\\
    		&=e^{i\pi (2k+1)\frac{r}{m}}\sum_{l=0}^{m-1}e^{i\pi (2k+1)\frac{l}{m}}\\
    		&=e^{i\pi (2k+1)\frac{r}{m}}\frac{2}{1-e^{i\pi\frac{2k+1}{m}}}\\
    		&=\frac{e^{i\frac{\pi}{2m}\left(m+(2r-1)(2k+1)\right)}}{\sin(\frac{(2k+1)\pi}{2m})}.
    	\end{aligned}
    \end{equation}
    Thus, we get
    \begin{equation}
    	\norm{\bm v_m(\bm a)}_\infty=\frac{1}{\sin(\frac{\pi}{2m})},
    \end{equation}
    and the $n$-norm is determined by
    \begin{equation}
    	\norm{\bm v_m(\bm a)}_n^n=\sum_{l=0}^{m-1}\frac{1}{\sin[n]((2l+1)\frac{\pi}{2m})}.
    \end{equation}
    As argued before these strategies maximize the $\infty$-norm and the $n$-norm does not depend on which strategy we choose from this set.

    \begin{figure}[htbp]
	    \centering
	    \begin{minipage}{0.45\linewidth}
	        \centering
	        \includegraphics[width=\linewidth]{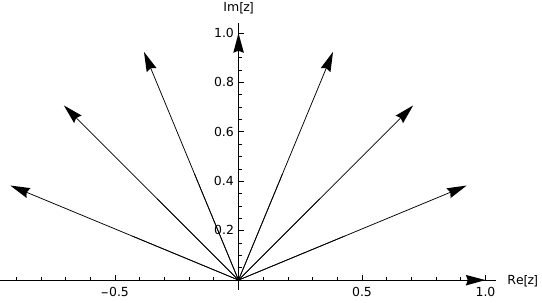}
	        \smallskip
	        \textbf{(a)} Unit vectors contributing to $v^0_m(\bm a)$ for $m=8$ measurements and $a_l=1$ for all $l\in\{0,\dots, m-1\}$.
	        \label{fig:unit_vec}
	    \end{minipage}
	    \hfill
	    \begin{minipage}{0.45\linewidth}
	        \centering
	        \includegraphics[width=\linewidth]{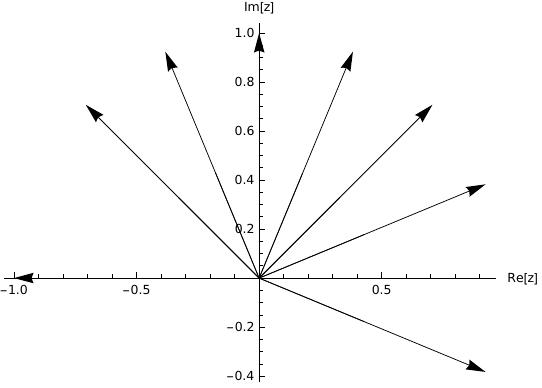}
	        \smallskip
	        \textbf{(b)} Unit vectors contributing to $v^0_m(\bm b)$ for $m=8$ measurements with $\norm{v_m(\bm b)}_\infty=\frac{1}{\sin{\frac{\pi}{2m}}}-4\sin(\frac{\pi}{2m})$.
	        \label{fig:unit_vec_2}
	    \end{minipage}
	    \caption{Comparison of the unit vector configurations.}
    \end{figure}
    
    Lastly, we consider the case of $\norm{\bm v_m(\bm b)}_\infty<\frac{1}{\sin{\frac{\pi}{2m}}}$. This appears if there is no half-plane containing all $m$ unit vectors. First, we show that at most $m-2$ unit vectors lying in every half-plane for every component $v^k_m(\bm b)$ cannot lead to a maximal $\norm{\bm v_m(\bm b)}_\infty$. Let $w$ be a complex number defined by one of the unit vectors in the sum of the component $v^k_m(\bm b)= \sum_{l=0}^{m-1}  b^{(j)}_{l} e^{i\pi (2k+1)\frac{l}{m}}$ not being in the half-plane
    \begin{equation}
    	H=\{x\in\mathbb{C}\; \vert\;\Re(x)\Re(v_m^k(\bm b))+\Im(x)\Im(v_m^k(\bm b))\geq0\}.
    \end{equation}
    Flipping the sign of $w$, i.e.\ adding $-2w$, increases the absolute value of the component $v_m^k(\bm b)$:
    \begin{equation}\label{eq:vector_sign_flip}
    	\abs{v_m^k(\bm b)-2w}^2= \abs{v_m^k(\bm b)}^2+4\abs{w}^2-4\underbrace{\left(\Re(v_m^k(\bm b)\Re(w)) +\Im(v_m^k(\bm b)\Im(w))\right)}_{<0}>\abs{v_m^k(\bm b)}^2.
    \end{equation}
    and there are at most $m-1$ vectors in every possible half-plane. Thus, the optimal choice in this case has to be one vector being flipped starting from all vectors being in one half-plane. If one of the outermost vectors is flipped, all vectors are in one half-plane again. Thus, we need to flip an inner vector. According to \cref{eq:vector_sign_flip}, the optimum is achieved for the minimal value of $\Re(v_m^k(\bm b)\Re(w))+\Im(v_m^k(\bm b)\Im(w))$, i.e.\ for a second-outermost vector being flipped. A quick calculation verifies that this leads to the norm
    \begin{align}
    	\norm{v_m(\bm b)-2w}_\infty=\frac{1}{\sin(\frac{\pi}{2m})}-4\sin(\frac{\pi}{2m}).
    \end{align}
    marking an upper bound for all strategies $\bm b$ with $\norm{\bm v_m(\bm b)}_\infty<\frac{1}{\sin(\frac{\pi}{2m})}$.

\end{proof}
Now, we are able to prove the local bounds of the inequalities. The next two propositions analyze the cases for even and odd $n$ separately.

\begin{prop}\label{prop:bound_even}
    Given any number of measurement settings $m=2^x$, the local bound is
    \begin{equation}\label{eq:opt_even}
        C_{m,n}=\frac 2{m^n}\sum_{j\in[\frac m2]}\sin[-n](\frac{(2j+1)\pi}{2m})
    \end{equation}
    for every even number of parties $n \geq n_m^{\text{even}}$ where
    \begin{equation}
        n_m^{\text{even}}= \left\lceil2\left(1+\frac{\log(\sqrt2)-\log(m\sin(\frac\pi{2m}))}{\log(1-4\sin[2](\frac{\pi}{2m}))}\right)\right\rceil.
    \end{equation}
\end{prop}
\begin{proof}
    Using \Cref{cor:lb_norm} and the fact that $v^j_m(\bm a)=v^{m-1-j}_m(\bm a)^\ast$, we can rewrite the local bound for even $n$ as
    \begin{equation}
        C_{m,n}=\frac2{m^n} \max_{\bm a}\norm{\bm v^+_{m}(\bm a)}_n^n.
    \end{equation}
    with $\bm v^+_{m}(\bm a)$ being the vector containing the first $\frac m2$ components of $\bm v_{m}(\bm a)$.
    As discussed in \Cref{lem:max_norm}, a strategy $\bm a$ with $\norm{\bm v_m(\bm a)}_\infty=\frac{1}{\sin(\frac{\pi}{2m})}$ leads to $\norm{\bm v^+_{m}(\bm a)}_n^n=\sum_{j\in[\frac m2]}\sin[-n](\frac{(2j+1)\pi}{2m})$.
    We are now going to show that this is optimal.

    \medskip
    For $m=4$, it can directly be checked that all strategies $\bm a\in\{\pm1\}^m$ have maximal $\norm{\bm v_m(\bm a)}_\infty=\frac{1}{\sin(\frac{\pi}{2m})}$, thus also maximal $\norm{\bm v_m(\bm a)}_n$.
    This proves the claim for all even $n$.
    Hence assume $m\geq 8$ from now on.
    We can upper-bound the $n$-norm of a vector $\bm k$ by interpolating between $n=2$ and $n=\infty$ as 
    \begin{equation}
        \norm{\bm k}_n\leq \norm{\bm k}_2^{\frac2n}\norm{\bm k}_\infty^{1-\frac2n}.
    \end{equation}
    From \Cref{lem:2norm} we know that $\lVert \bm v_m(\bm a) \rVert_2 = m$ for every vector $\bm v_m(\bm a)$, and from \Cref{lem:max_norm} we know that if $\norm{\bm v_m(\bm b)}_\infty$ is sub-optimal, then $\norm{\bm v_m(\bm b)}_\infty \leq \frac{1}{\sin(\frac{\pi}{2m})} - 4 \sin \left( \frac{ \pi }{2m} \right)$. This also holds for $\bm v^+_m(\bm a)$ and by substituting into the norm interpolation equation, we get
    \begin{equation}\label{eq:opt_n_norm}
        \norm{\bm v_m^+(\bm b)}_n\leq \left(\frac{m}{\sqrt2}\right)^{\frac{2}{n}}\left(\frac{1}{\sin(\frac{\pi}{2m})}-4\sin(\frac{\pi}{2m})\right)^{1-\frac{2}{n}}.
    \end{equation}
    Direct calculation shows that this is less than $\frac{1}{\sin(\frac{\pi}{2m})}$ for all $n \geq n_m^{\text{even}}$ with
    \begin{equation}
        n_m^{\text{even}}= \left\lceil2\left(1+\frac{\log(\sqrt2)-\log(m\sin(\frac\pi{2m}))}{\log(1-4\sin[2](\frac{\pi}{2m}))}\right)\right\rceil.
    \end{equation}
    Therefore, for $ n \ge n_m^{\text{even}}$, if $\bm b$ is a sub-optimal strategy in $\infty$-norm, it has $\lVert \bm v_m^+(\bm b)\rVert_n < \frac{1}{\sin(\pi/2m)}$. On the other hand, if $\bm a$ is optimal in the $\infty$-norm, by monotonicity of the norm we get $\lVert \bm v_m^+(\bm a) \rVert_n \ge \lVert\bm  v_m^+(\bm a) \rVert_\infty = \frac{1}{\sin(\pi/2m)}$.
    Altogether, this proves no sub-optimal $\infty$-norm strategy can maximize the $n$-norm, and any $n$-norm maximizer also maximizes the $\infty$-norm. Thus,
    \begin{equation}
        C_{m,n}=\frac 1{m^n}\max_{\bm a} \norm{\bm v_m(\bm a)}_n^n=\frac 2{m^n}\max_{\bm a} \norm{\bm v^+_m(\bm a)}_n^n=\frac 2{m^n}\sum_{j\in[\frac m2]}\sin[-n](\frac{(2j+1)\pi}{2m}), \quad \forall n \geq n_m^{\text{even}} .
    \end{equation}
\end{proof}

\begin{prop}\label{prop:bound_odd}
	Given any number of measurement settings $m=2^x$, the local bound is
	\begin{equation}
		C_{m,n}=\frac 2{m^n}\sum_{j\in[\frac m2]}(-1)^j\cos(\frac{(2j+1)\pi}{2m})\sin[-n](\frac{(2j+1)\pi}{2m}).
	\end{equation}
	for all odd $n\geq n_m^{\text{odd}}=\max\{n_m^{\text{even}}+1,n_m^{\text{odd},1},n_m^{\text{odd},2}\}$ with
	\begin{equation}
		n_m^{\text{odd},1}=\left\lceil\frac{\log(16)-\log(5m^2)}{\log(\tan(\frac{\pi}{8}))}\right\rceil, \qquad n_m^{\text{odd},2}=\left\lceil\frac{\log(1- \cos(\frac{\pi}{2m}))- \log(1 +\frac54 \cos(\frac{\pi}{2m}))}{\log(\tan(\frac{\pi}{8}))}\right\rceil.
	\end{equation}
\end{prop}

\begin{proof}
	This proof assumes that $n>n_m^{\text{even}}$. First, we bound the expectation value of the inequality for a set of strategies $\{\bm a^{(i)}\}$ where at least one party $i$ chooses a strategy $\bm a^{(i)}$ with $\norm{\bm v_m(\bm a^{(i)})}_\infty<\frac{1}{\sin(\frac{\pi}{2m})}$. Then,
		\begin{equation}\label{eq:ub_mixed}
			\begin{aligned}
				\expval{\MM_{m,n}}_{\bm a^{(1)},\ldots, \bm a^{(n)}}
				&\leq \frac1{m^n} \norm{\bm v_m(\bm a^{(i)})}_\infty\max_{\bm b}\norm{\bm v_m(\bm b)}_{n-1}^{n-1} \\
				&\leq \frac1{m^n} \left(\frac{1}{\sin(\frac{\pi}{2m})}-4\sin(\frac{\pi}{2m})\right)\max_{\bm b}\norm{\bm v_m(\bm b)}_{n-1}^{n-1} \\
				&\leq \frac1{m^n} \frac{\cos[2](\frac{\pi}{2m})}{\sin(\frac{\pi}{2m})}\max_{\bm b}\norm{\bm v_m(\bm b)}_{n-1}^{n-1}.
			\end{aligned}
		\end{equation}
	Here the last inequality follows because taking the same strategy for every party, up to complex conjugation, is optimal in the case of an even number of parties.
	
	If every party chooses a strategy $\bm a^{(i)}$ with $\norm{\bm v_m(\bm a^{(i)})}_\infty=\frac{1}{\sin(\frac{\pi}{2m})}$, but the positions of the entries with the greatest absolute value do not coincide, we can upper-bound the expectation value as follows. Without loss of generality we assume that $\abs{v^{0}_m(\bm a^{(1)})}=\frac{1}{\sin(\frac{\pi}{2m})}$ and thus there is a party $j$ with  $\abs{v^0_m(\bm a^{(j)})}\leq\frac{1}{\sin(\frac{3\pi}{2m})}\leq \tan(\frac{\pi}{8})\frac{1}{\sin(\frac{\pi}{2m})}\leq\frac12\frac{1}{\sin(\frac{\pi}{2m})}$ for $m\geq4$. Then,
	\begin{equation}\label{eq:ub_big_mixed}
		\begin{aligned}
			\expval{\MM_{m,n}}_{\bm a^{(1)},\ldots, \bm a^{(n)}}&=\frac2{m^n}\sum_{k=0}^{\frac{m}{2}-1}\Re\left[v_m^{k}(\bm a^{(1)})\prod_{j=2}^n v_m^{k}(\bm a^{(j)})\right]\\
            &\leq \frac2{m^n}\sum_{k=0}^{\frac{m}{2}-1}\abs{v_m^{k}(\bm a^{(1)})}\prod_{j=2}^n \abs{v_m^{k}(\bm a^{(j)})}\\
			& =\frac2{m^n}\left(\abs{v_m^{0}(\bm a^{(1)})}\prod_{j=2}^n \abs{v_m^{0}(\bm a^{(j)})} +\sum_{k=1}^{\frac{m}{2}-1}\abs{v_m^{k}(\bm a^{(1)})}\prod_{j=2}^n \abs{v_m^{k}(\bm a^{(j)})}\right)\\
			& \leq\frac2{m^n}\left(\left(\frac{1}{\sin(\frac{\pi}{2m})}-\frac{1}{\sin(\frac{3\pi}{2m})}\right)\frac12\prod_{j=2}^n \frac{1}{\sin(\frac{\pi}{2m})} +\frac{1}{\sin(\frac{3\pi}{2m})}\prod_{j=2}^n \norm{\bm v_m^+(\bm a^{(j)})}_n\right)\\
			  &\leq \frac2{m^n} \left(\left( \frac{1}{\sin(\frac{\pi}{2m})}-\frac{1}{\sin(\frac{3\pi}{2m})}\right)\frac12 \max_{\bm b}\norm{\bm v_m^+(\bm b)}_{n-1}^{n-1}+\frac{1}{\sin(\frac{3\pi}{2m})}\max_{\bm b}\norm{\bm v_m^+(\bm b)}_{n-1}^{n-1}\right)\\
			&=\frac1{2m^n} \left( \frac{1}{\sin(\frac{\pi}{2m})}+\frac{1}{\sin(\frac{3\pi}{2m})}\right)\max_{\bm b}\norm{\bm v_m(\bm b)}_{n-1}^{n-1} \\
            &\leq \frac1{m^n} \frac{1}{\sqrt2} \frac{1}{\sin(\frac{\pi}{2m})}\max_{\bm b}\norm{\bm v_m(\bm b)}_{n-1}^{n-1}.
		\end{aligned}
	\end{equation}
	
	It is better to apply strategies with $\norm{\bm v_m(\bm a^{(i)})}_\infty=\frac{1}{\sin(\frac{\pi}{2m})}$ with the absolute values of the components of $\bm v_m(\bm a^{(i)})$ being in the same order. We know the components of each vector $\bm v_m(\bm a^{(i)})$ from \cref{eq:max_comp} and the phases in the product are from the same set as the phases of the single strategies:
	\begin{equation}
	\begin{aligned}
		\expval{\MM_{m,n}}_{\bm a^{(1)},\ldots,\bm a^{(n)}}&=\frac2{m^n} \sum_{k=0}^{\frac m2 -1}\frac{\cos(\frac{n\pi}{2}+(2\sum_{i=1}^n r^i-n)(2k+1)\frac{\pi}{2m})}{\sin[n]((2k+1)\frac{\pi}{2m})}\\
		&=\frac2{m^n} \sum_{k=0}^{\frac m2 -1}\frac{\cos(\frac{\pi}{2}+(2r^\prime-1)(2k+1)\frac{\pi}{2m})}{\sin[n]((2k+1)\frac{\pi}{2m})}
	\end{aligned}
	\end{equation}
	with
	\begin{equation}
		r^\prime=\left[\sum_{i=1}^{n}r^i+\frac{n-1}{2}(m-1)\right]_{2m}.
	\end{equation}
	The term with $k=0$ is dominant and we can show that it is optimal to minimize the phase of the cosine with $r^\prime\equiv-\frac m2\pmod{2m}$, because the other terms are exponentially smaller for growing $n$:
	\begin{equation}\label{eq:other_terms}
	\begin{aligned}
		\frac2{m^n}\sum_{k=1}^{\frac m2 -1}\frac{\cos(\frac{\pi}{2}+(2r^\prime-1)(2k+1)\frac{\pi}{2m})}{\sin[n]((2k+1)\frac{\pi}{2m})}
		&\leq\frac2{m^n}\sum_{k=1}^{\frac m2 -1}\frac{1}{\sin[n]((2k+1)\frac{\pi}{2m})} \\
		&\leq\frac2{m^n}\sum_{k=1}^{\frac m2 -1}\left(\frac{1}{\sin(\frac{\pi}{2m})}\frac{3\gamma}{2k+1}\right)^n\\
		&\leq\frac2{m^n}\frac{(3\gamma)^n}{\sin[n](\frac{\pi}{2m})} \left(\frac1{3^n}+\int_{1}^{\frac m2 -1}\frac1{(2k+1)^n}\dd k\right)\\
		&=\frac2{m^n}\frac{(3\gamma)^n}{\sin[n](\frac{\pi}{2m})} \left(\frac1{3^n}+\frac1{2(n-1)}\left(\frac{1}{3^n}-\frac{1}{(m-1)^n}\right)\right)\\
		&\leq\frac2{m^n} \frac{\gamma^n}{\sin[n](\frac{\pi}{2m})} \left(1+\frac1{2(n-1)}\right)\leq\frac5{2m^n}\frac{\gamma^n}{\sin[n](\frac{\pi}{2m})} ,
	\end{aligned}
	\end{equation}
	with $\gamma=\tan(\frac{\pi}{8})$ and the last inequality holding for $n\geq3$. A suboptimal choice of the cosine's phase for the $k=0$ term leads to a difference of at least
	\begin{equation}
		\frac2{m^n}\left(\frac{\cos(\frac{\pi}{2m})}{\sin[n](\frac{\pi}{2m})}-\frac{\cos(\frac{3\pi}{2m})}{\sin[n](\frac{\pi}{2m})}\right)=\frac2{m^n}\frac{2}{\sin[n](\frac{\pi}{2m})}\sin(\frac{\pi}{2m})\sin(\frac{\pi}{m})\geq \frac2{m^n}\frac{2}{\sin[n](\frac{\pi}{2m})}\frac1m\frac{2}{m}= \frac{8}{m^{n+2}}\frac{1}{\sin[n](\frac{\pi}{2m})}
	\end{equation}
	in this term. This is greater than the contribution of all other terms bounded in \cref{eq:other_terms} for all sufficiently large $n$:
	\begin{equation}
	\begin{aligned}
		\frac5{2m^n}\frac{\gamma^n}{\sin[n](\frac{\pi}{2m})}&\leq\frac{8}{m^{n+2}}\frac{1}{\sin[n](\frac{\pi}{2m})} \\
		\Longleftrightarrow\quad  \gamma^n&\leq \frac{16}{5m^2}\\
		\Longleftrightarrow n\log(\gamma)&\leq \log(16)-\log(5m^2)\\
		\Longleftrightarrow n &\geq \frac{\log(16)-\log(5m^2)}{\log(\gamma)},
	\end{aligned}
	\end{equation}
	i.e. the choice $r^\prime\equiv-\frac{m}{2}\pmod{2m}$ is optimal for all $n\geq n_m^{\text{odd},1}=\left\lceil\frac{\log(16)-\log(5m^2)}{\log(\gamma)}\right\rceil$. This directly leads to
	\begin{equation}\label{eq:odd_opt_candidate}
		\expval{\MM_{m,n}}_{\bm a^{(1)},\ldots,\bm a^{(n)}}=\frac 2{m^n}\sum_{j\in[\frac m2]}(-1)^j\cos(\frac{(2j+1)\pi}{2m})\sin[-n](\frac{(2j+1)\pi}{2m})
	\end{equation}
	being the optimal value for strategies with $\norm{\bm v_m(\bm a^{(i)})}_\infty=\frac{1}{\sin(\frac{\pi}{2m})}$ and coinciding positions of the greatest entries for all $n\geq n_m^{\text{odd},1}$. This can be lower-bounded as
	\begin{equation}\label{eq:lb_best_strat}
		\expval{\MM_{m,n}}_{\bm a^{(1)},\ldots, \bm a^{(n)}}\geq \frac2{m^n} \left(\frac{\cos(\frac{\pi}{2m})}{\sin[n](\frac{\pi}{2m})}-\frac{\cos(\frac{3\pi}{2m})}{\sin[n](\frac{3\pi}{2m})}\right)\geq \frac2{m^n} \frac{\cos(\frac{\pi}{2m})}{\sin[n](\frac{\pi}{2m})}\left(1-\gamma^n\right).
	\end{equation}
	We compare this to the upper bounds derived for the other cases in \cref{eq:ub_mixed,eq:ub_big_mixed}. The upper bound in \cref{eq:ub_mixed} is greater than the one in \cref{eq:ub_big_mixed} and we use the inequality $\max_{\bm a}\norm{\bm v_m(\bm a)}_n^n\leq 2\left(1+\frac54 \gamma^n\right)\frac{1}{\sin^n(\frac{\pi}{2m})}$ --- which can be derived with \cref{eq:other_terms} --- to further bound \cref{eq:ub_mixed}. Then we compare this to the lower bound in \cref{eq:lb_best_strat}:
	\begin{equation}
	\begin{aligned}
		\frac2{m^n} \frac{\cos(\frac{\pi}{2m})}{\sin[n](\frac{\pi}{2m})}\left(1-\gamma^n\right)&\geq\frac2{m^n} \cos[2](\frac{\pi}{2m})\left(1+\frac54 \gamma^n\right)\frac{1}{\sin[n](\frac{\pi}{2m})}\\
		\Longleftrightarrow 1-\gamma^n&\geq \cos(\frac{\pi}{2m})\left(1+\frac54 \gamma^n\right)\\
		\Longleftrightarrow 1- \cos(\frac{\pi}{2m}) &\geq \left(1+\frac54\cos(\frac{\pi}{2m})\right)\gamma^n\\
		\Longleftrightarrow \log(1- \cos(\frac{\pi}{2m}))- \log(1 +\frac54 \cos(\frac{\pi}{2m}))&\geq n\log(\gamma)\\
		\Longleftrightarrow n &\geq \frac{\log(1- \cos(\frac{\pi}{2m}))- \log(1 +\frac54 \cos(\frac{\pi}{2m}))}{\log(\gamma)}.
	\end{aligned}
	\end{equation}
    Hence, for $n\geq n_m^{\text{odd},2}=\left\lceil\frac{\log(1- \cos(\frac{\pi}{2m}))- \log(1 +\frac54 \cos(\frac{\pi}{2m}))}{\log(\gamma)}\right\rceil$ the strategy presented in \cref{eq:odd_opt_candidate} performs better than all strategies which do not lead to a coordinate $k$ with $\abs{v^k_m(\bm a^{(i)})}=\frac{1}{\sin(\frac{\pi}{2m})}$ for all parties $i$.
	We can summarize the results as follows: the expectation value in \cref{eq:odd_opt_candidate} is the local bound for all odd $n\geq n_m^{\text{odd}}=\max\{n_m^{\text{even}} +1,n_m^{\text{odd},1},n_m^{\text{odd},2}\}$.
\end{proof}
Note that $n_m^{\text{even}}$ and $n_m^{\text{odd}}$ are not tight bounds for the given local bounds to hold. They are the result of finding an analytical solution valid for all $m=2^x$. Investigating specific numbers of measurements $m$ can improve the minimal number of parties for the bounds to hold. We conjecture that they hold for all $m=2^x$ and $n\in\mathbb{N}$ as we could not find any counterexample.

\subsubsection{Considered measurement settings}
Here, we show that the stated bounds hold for the cases considered in the paper including the bounds needed for nonlocality depth certification. They are $m=8$ for $n\in\{3,4,5,6,7,8,16,28,43,60,80\}$ as well as $m\in\{4,16,32\}$ for $n=80$. All of the cases with even $n$ are covered by \Cref{prop:bound_even} as $n_8^{\text{even}}=4$, $n_{16}^{\text{even}} = 8$ and $n_{32}^{\text{even}} = 24$. For the odd numbers of parties we have $n_8^{\text{odd},1}=4$ and $n_8^{\text{odd},2}=6$ such that $n\geq7$ is covered by \Cref{prop:bound_odd}. We verified the local bound for $m=8$ and $n=3$ numerically. We cover the case of $m=8$ and $n=5$ by explicitly computing the upper bound in Eq.~\eqref{eq:ub_mixed} and directly comparing it to the best set of local strategies with $\norm{\bm v_m(\bm a^{(i)})}_\infty=\frac{1}{\sin(\frac{\pi}{2m})}$ for all parties $i\in\{1,\ldots,5\}$.

\subsubsection{Scaling in comparison to Mermin Inequalities}

The improved scaling of the local bounds --- and hence the noise robustness --- with the number of parties $n$ is an essential feature of the generalized Mermin inequalities. The local bound of Mermin's inequalities $\MM_{2,m}$ scale as $C_{2,m}\propto(1/\sqrt2)^n\approx0.707^n$. In contrast, our inequalities achieve local bounds $C_{m,n}\propto 1/(m\sin(\frac{\pi}{2m}))^n$ in the limit of large $n$, e.g.\ for $m=8$ measurement settings $C_{8,n}\propto 0.641^n$. Taking the limit of infinitely many measurement settings, we get a scaling of $C_{m\to\infty,n}\propto(2/\pi)^n\approx 0.637^n$ also reported in \cite{sm:zukowski1993bell} for a continuous choice of measurement settings.

To the best of our knowledge, the generalized Mermin inequalities have the best noise robustness for GHZ states compared to all other known Bell inequalities in the scenarios of $m\geq16$ dichotomic measurements on an odd number of parties $n\geq n^{\text{odd}}_m$.

\subsection{Nonlocality depth}
\label{sec:nonlocality-depth}

We characterize the nonlocality depth by investigating its robustness under the addition of nonlocal resources, i.e.\ classical communication in $l$ subgroups of the $n$ parties, following the definitions in Ref.~\cite{sm:bancal2009quantifying}.
In the \emph{grouping model} (\Cref{fig:grouping-model}), every party belongs to one of the $l$ groups and arbitrary communication and collaboration is allowed in each group, but none between distinct groups.
To derive the nonlocality depth inequalities, we first consider the less restricted \emph{restrained subset model} (\Cref{fig:restrained-subset-model}), which describes a superset of the correlations possible in the grouping model, and later translate the results to the grouping model. 
In the restrained subset model, there is a subset $S$ of $l$ parties such that, after all allowed communication among the $n$ parties, the information available to any party before producing its output contains the input of at most one party from $S$.
The remaining parties may otherwise communicate and collaborate arbitrarily, subject to this restriction.
In particular, since each party in $S$ already has access to its own input, no party in $S$ can learn the input of any other party in $S$.

\smallskip
The grouping model is a subset of the restrained model in the sense that every strategy in the grouping model with $l$ groups is also a valid strategy in the restrained subset model: choosing one representative party from each group gives a subset $S$ whose inputs cannot be combined at any site, since no communication is allowed between distinct groups.

\begin{figure}[htbp]
    \centering

    \begin{minipage}[t]{0.208\linewidth}
        \centering
        \includegraphics[width=\linewidth]{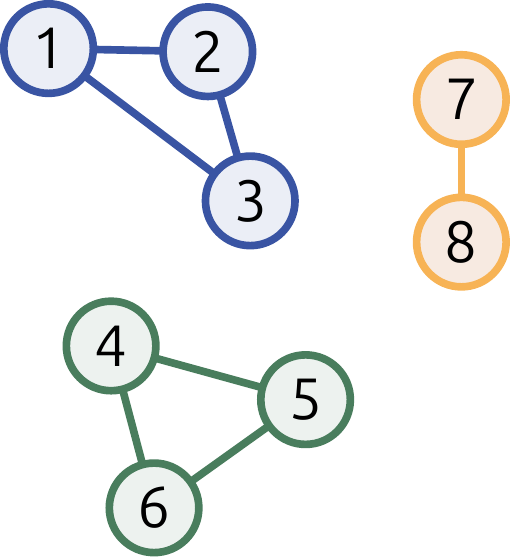}
        \smallskip
        \textbf{(a)} Grouping model
        \label{fig:grouping-model}
    \end{minipage}
    \hspace{4.5em}
    \begin{minipage}[t]{0.23\linewidth}
        \centering
        \includegraphics[width=\linewidth]{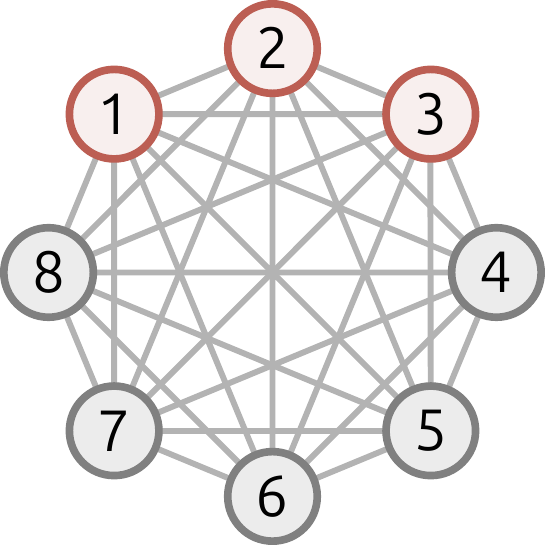}
        \smallskip
        \textbf{(b)} Restrained subset model
        \label{fig:restrained-subset-model}
    \end{minipage}

    \caption{
    Schematic representation of the models used to certify nonlocality depth.
    (a) In the grouping model, the $n$ parties are partitioned into groups. Parties within the same group may communicate and collaborate arbitrarily, while no communication is allowed between distinct groups.
    (b) In the restrained subset model, a subset $S$ of $l$ parties (highlighted in red) is chosen such that, after all allowed communication, no party has access to more than one input from $S$. The remaining parties may otherwise communicate freely, subject to this restriction.
    \Cref{prop:nonlocal_depth} and \Cref{cor:nonlocal_depth} show that correlations generated by such a grouping model satisfy $\langle \mathcal M_{m,n}\rangle \le C_{m,\lceil n/p\rceil}$. Observing a violation of this bound rules out all models with group size at most $p$, therefore certifies nonlocality depth larger than $p$.}
    \label{fig:nonlocality-depth-models}
\end{figure}

\begin{prop}\label{prop:nonlocal_depth}
	The maximum achievable value for $\MM_{m,n}$ in the restrained subset model with a subset of $l$ parties is given by
	\begin{equation}
		\max_{l\text{\normalfont -subset}} \expval{\MM_{m,n}}=C_{m,l}
	\end{equation}
	This can also be achieved with a grouping model with $l$ groups.
\end{prop}
\begin{proof}
	We recursively define the expressions
	\begin{equation}\label{eq:rec_def}
		M_{m,n}^k=\frac1m\left(\sum_{j=0}^k A_j^{(n)} M_{m,n-1}^{k-j}-\sum_{j=k+1}^{m-1} A_j^{(n)} M_{m,n-1}^{k-j+m}\right)
	\end{equation}
	with $M_{m,1}^k=A_k^{(1)}$ for all $k\in\{0,\ldots,m-1\}$ and the second sum vanishing if $k+1>m-1$. Note that for $k=0$ we recover the previously presented Generalized Mermin inequalities: $ \MM_{m,n}=M_{m,n}^0 $.

    \smallskip
	We decompose $M_{m,n}^0$ for all but the first $l$ parties, using the recursive definition in \cref{eq:rec_def}:
	\begin{equation}
		M_{m,n}^0=\frac1{m^{n-l}}\sum_{k_n,\ldots,k_{l+1}=0}^{m-1}A_{k_n}^n\ldots A_{k_{l+1}}^{l+1}\mathfrak{M}_{m,l}^{k_n,\ldots,k_{l+1}}
	\end{equation}
	with $\mathfrak{M}_{m,l}^{k_n,\ldots,k_{l+1}}\in\{\pm M_{m,l}^k\,|\, k\in\{0,\ldots,m-1\} \}$. This can be turned into an effective inequality on the first $l$ parties as presented in Ref.~\cite{sm:bancal2009quantifying} for two-setting inequalities. 
	Thus, the maximal value in the restrained subset model is upper-bounded by
	\begin{equation}\label{eq:restrained_model_bound}
		\max_{l\text{-subset}} \expval{M_{m,n}^0}\leq C_{m,l}.
	\end{equation}
	The upper bound can be reached even in a grouping model with $l$ groups $G_i$, where each of the first $l$ parties $i\in G_i$ belongs to a distinct group. The parties in each group $G_i$ share their inputs with party $i$. The parties $i\leq l$ compute the sum of the inputs $\tilde k_i=\left[\sum_{x\in G_i}k_x\right]_{2m}$ and outputs
	\begin{equation}
		a^{(i)}_{k_i}=\begin{cases}
			b^{(i)}_{[\tilde k_i]_m}, & \tilde k_i<m\\ -b^{(i)}_{[\tilde k_i]_m}, & \tilde k_i\geq m\\
		\end{cases},
	\end{equation}
	where $\{b^{(i)}\}_{i=1}^l$ is an optimal strategy for the $l$-partite expression $\MM_{m,l}$. Every other party returns the output $+1$ for every input. Then, the upper bound in \cref{eq:restrained_model_bound} is saturated, proving the claim.
\end{proof}
This bound matches the one for the Mermin-Svetlichny inequalities used in Ref.~\cite{sm:bancal2009quantifying} for $l=3$ and performs better for all $l\geq 4$ (only one exception: $l=4$ and $m=4$). Note that the bound is trivial for the cases of $l=2$.
\smallskip
As the number of groups is limited by the size of the largest group, the previous proposition also gives the bounds for \emph{$p$-producible correlations}, that is, correlations generated by groupings in which every group has size at most $p$:
\begin{cor}
	The maximal achievable expectation value of $\MM_{m,n}$ in a grouping model with maximal group size $p$ is given by
	\begin{equation}
		\max_{p- \text{prod.}} \expval{\MM_{m,n}} = C_{m,\lceil \frac n p\rceil}.
	\end{equation}
\label{cor:nonlocal_depth}
\end{cor}
\begin{proof}
Any grouping with largest group size $p$ must have at least $\lceil n/p\rceil$ groups.
Choosing one representative from each of these groups gives a restrained subset of size $\lceil n/p\rceil$, therefore by \Cref{prop:nonlocal_depth} the Generalized Mermin expression is upper bounded by \(C_{m,\lceil n/p\rceil}\).
For the other direction, we can partition the $n$ parties into exactly $\lceil n/p\rceil$ groups, each of size at most $p$, and the construction in \Cref{prop:nonlocal_depth} attains this bound.
\end{proof}

\section{Additional GHZ Benchmarks}
\label{sec:si_ghz_benchmarks}

\subsection{Biseparable comparison at \texorpdfstring{$n=16$}{n=16}}

The smaller-scale GME benchmark discussed in the main text compares the Generalized Mermin correlator with the biseparable threshold
\begin{equation}
C_{\mathrm{bisep}}(m)
=
\frac{1}{m\sin(\pi/2m)} .
\end{equation}
This threshold follows from the multisetting DIEW/extended-parity-game biseparable quantum bound of P\'al and V\'ertesi~\cite{sm:pal2011multisetting}, expressed in the normalization used in the main text. In our notation, the corresponding unnormalized $n$-party expression has ideal GHZ/algebraic value
\begin{equation}
\widetilde Q_{m,n}=m^{n-1},
\end{equation}
whereas the normalized operator used throughout the main text,
\begin{equation}
\MM_{m,n}=\frac{\widetilde{\MM}_{m,n}}{m^{n-1}},
\end{equation}
has ideal value
\begin{equation}
Q_{m,n}
=
\bra{\mathrm{GHZ}_n}
\MM_{m,n}
\ket{\mathrm{GHZ}_n}
=
1 .
\end{equation}
In this normalized convention, the P\'al--V\'ertesi biseparable quantum bound gives
\begin{equation}
\langle \MM_{m,n}\rangle_{\mathrm{bisep}}
\le
C_{\mathrm{bisep}}(m)
=
\frac{1}{m\sin(\pi/2m)} .
\end{equation}
This is a Bell/DIEW-style biseparable quantum bound and does not rely on trusting the specific Pauli $X/Y$ measurement axes. It is weaker than the fixed-measurement coherence-witness bound $\langle \MM_{m,n}\rangle_{\mathrm{bisep}}\le 1/2$, but it has the appropriate biseparable-quantum Bell-benchmark interpretation used in the main text.

Using this biseparable quantum benchmark, for the $16$-qubit GHZ state we obtain the comparison summarized in Table~\ref{tab:gme16}.
The Generalized Mermin test improves the gap to the relevant biseparable threshold by almost a factor of two.
We therefore treat this comparison as a secondary GME-oriented benchmark; it is not intended as a loophole-free or fully device-independent GME claim.

\begin{table}[ht]
\centering
\begin{tabular}{lccc}
\toprule
Setting number $m$ & $\langle \MM_{m,16}\rangle_{\exp}$ & $C_{\mathrm{bisep}}(m)$ & Gap \\
\midrule
2 & 0.7446 & 0.7071 & 0.0375 \\
8 & 0.7128 & 0.6407 & 0.0721 \\
\bottomrule
\end{tabular}
\caption{Biseparable comparison for the $16$-qubit GHZ state. The generalized-Mermin test yields a noticeably larger normalized margin above the biseparable threshold.}
\label{tab:gme16}
\end{table}

\subsection{Detailed nonlocality-depth hierarchy}
Table~\ref{tab:depth_summary_detailed} gives the full grouping-model analysis used for the nonlocality-depth claims.
For each pair $(n,m)$ we list the observed correlator, the threshold for $(k-1)$-producible models, the certified depth $k$, and the settings-level EB $p$-value upper bound.
The same table also includes the fixed-size $80$-qubit scan over $m=4,16,32$.

\begin{table}[!htbp]
\centering
\begin{small}
\begin{tabular}{ccccccc}
\toprule
$n$ & $m$ & $\langle \MM_{m,n}\rangle_{\exp}$ & $C_{(k-1)\text{-prod}}$ & Depth $k$ & $N\times M$ & $\log_{10}(p_{\mathrm{EB}})$ \\
\midrule
\multicolumn{7}{l}{\textit{Standard Mermin ($m=2$)}} \\
4  & 2 & 0.9364 & 0.5000 & 2  & $400\times1500$  & -15.711 \\
8  & 2 & 0.8680 & 0.5000 & 4  & $400\times1500$  & -13.122 \\
16 & 2 & 0.7446 & 0.5000 & 8  & $500\times1500$  & -10.710 \\
28 & 2 & 0.5521 & 0.5000 & 14 & $800\times1500$  & -3.285 \\
43 & 2 & 0.3714 & 0.2500 & 11 & $1300\times1500$ & -13.502 \\
60 & 2 & 0.2297 & 0.1250 & 10 & $2200\times1500$ & -19.713 \\
80 & 2 & 0.1012 & 0.0625 & 10 & $3600\times1500$ & -11.287 \\
\midrule
\multicolumn{7}{l}{\textit{Generalized Mermin ($m=8$)}} \\
4  & 8 & 0.9437 & 0.5000 & 2  & $400\times1500$  & -16.013 \\
8  & 8 & 0.8671 & 0.5000 & 4  & $400\times1500$  & -13.142 \\
16 & 8 & 0.7128 & 0.5000 & 8  & $500\times1500$  & -9.055 \\
28 & 8 & 0.5651 & 0.5000 & 14 & $800\times1500$  & -4.247 \\
43 & 8 & 0.3708 & 0.2109 & 11 & $1300\times1500$ & -18.055 \\
60 & 8 & 0.2189 & 0.1387 & 12 & $2200\times1500$ & -14.684 \\
80 & 8 & 0.1012 & 0.0869 & 14 & $3600\times1500$ & -3.659 \\
\midrule
\multicolumn{7}{l}{\textit{Fixed-size scan for the $80$-qubit GHZ state}} \\
80 & 4  & 0.1058 & 0.0937 & 14 & $3600\times1500$ & -3.028 \\
80 & 16 & 0.1082 & 0.0853 & 14 & $3600\times1500$ & -6.429 \\
80 & 32 & 0.1094 & 0.0849 & 14 & $3600\times1500$ & -6.978 \\
\bottomrule
\end{tabular}
\end{small}

\caption{Detailed nonlocality-depth hierarchy for GHZ states.
The threshold $C_{(k-1)\text{-prod}}$ is obtained from the analytical grouping-model bound.
The last column reports settings-level empirical-Bernstein $p$-value upper bounds on a $\log_{10}$ scale. For each $(n,m)$, the reported depth $k$ is the largest depth for which
the corresponding $(k-1)$-producible null hypothesis is rejected with
\(p_{\mathrm{EB}}\le 1\times 10^{-3}\), equivalently
\(\log_{10}(p_{\mathrm{EB}})\le -3\).
\label{tab:depth_summary_detailed}
}

\end{table}
\FloatBarrier

Two trends are worth highlighting. First, the generalized family and the standard Mermin test track the same depth hierarchy at small and intermediate sizes, as expected when the GHZ coherence is still high. Second, once phase noise becomes the dominant limitation, the generalized family retains a clear advantage: for example, the certified depth rises from $10$ to $14$ at $n=80$ when the Bell test is upgraded from $m=2$ to $m=8$.

\section{Discussion on Bell Loopholes and Experimental Constraints}
\label{sec:bell_loophole}
While our experiment demonstrates high-significance Bell violations at a scale of 80 qubits, it is important to clarify the status of the three primary Bell loopholes—locality, freedom-of-choice, and detection—within the context of our superconducting platform, which are central considerations in experimental tests of Bell inequalities~\cite{sm:Bell1964,sm:clauser1969proposed,sm:hensen2015loophole,sm:giustina2015significant}.

\subsection{Detection Loophole}
The detection loophole is associated with loss and post-selection: if a non-negligible fraction of trials yields no registered outcome and those trials are discarded, a local-realistic model can bias the retained sample~\cite{sm:eberhard1993background,sm:garg1987detector}. 
In the present superconducting-qubit experiment, every qubit is read out in every shot and no trials are discarded on the basis of missing detections. The recorded-outcome rate is therefore unity, and the conventional loss-based detection loophole associated with discarded no-click events is not operative here.

The reported state-assignment fidelity of $99.66\%$, as shown in \cref{fig:gr_errors}, quantifies the accuracy of outcome classification rather than the detection efficiency.
Imperfect state assignment can attenuate the measured correlators, but it does not introduce a sample-selection bias because all shots are retained. 

\subsection{Locality Loophole}
The locality loophole requires that the measurement events of different qubits be space-like separated~\cite{sm:weihs1998violation}, such that no signal (traveling at the speed of light $c$) can communicate the choice of measurement setting from one qubit to another~\cite{sm:shalm2015strong,sm:li2018test,sm:storz2023loophole}. In our integrated circuit, the physical distance between the most distant qubits is on the order of several $\text{cm}$, while the total duration of the measurement sequence (including equatorial rotations and readout) is in the range of $\mu\text{s}$. For a loophole-free test, the measurement must be completed in a time $t < d/c$. At $d \approx 3\text{cm}$, this would require a total measurement time of $\sim 100\text{ ps}$, which is several orders of magnitude faster than current superconducting control electronics. Therefore, the locality loophole remains open in this study, as is typical for integrated solid-state processors.

\subsection{Freedom-of-Choice Loophole}
The freedom-of-choice loophole assumes that the measurement settings might be influenced by a hidden variable shared with the state preparation~\cite{sm:giustina2015significant}. To close this, the settings must be chosen by independent random number generators (RNGs) that are space-like separated from the state-preparation event~\cite{sm:handsteiner2017cosmic,sm:rauch2018cosmic}. In our experiment, the measurement settings used for the Bell estimator were pseudo-randomly generated by software before execution and run in randomized order by an automated control script. This randomization is used to avoid systematically tying particular Bell terms to slow hardware drift. It is not an independent, space-like separated random-choice process. Since both the state preparation and the setting choices are implemented within the same cryostat and centralized control electronics, the freedom-of-choice loophole is not closed.

\subsection{Implications for Benchmarking}
Because our experiment does not simultaneously close all loopholes, it should not be interpreted as a loophole-free test of local realism. Rather, under the stated on-chip assumptions, the results provide a correlation-only Bell benchmark of large-scale GHZ correlations based on directly measured correlators and analytical bounds. The observed violations also reflect the processor's ability to preserve multipartite GHZ coherence, but this coherence interpretation is diagnostic rather than the basis of the Bell-statistical claims. For the present platform, the main open loopholes are locality and freedom-of-choice, whereas loss-based sample selection is not the operative limitation.  Although demonstrated here on superconducting hardware, the finite-setting GMI framework and its analytical bounds are platform independent and are, in principle, compatible with future implementations designed to close the standard Bell loopholes.

\bibliographystyle{apsrev4-2}

\end{document}